\documentclass[conference,a4paper]{IEEEtran}
\usepackage{amsmath,amssymb,mathrsfs}
\usepackage{epsfig,epsf,subfigure,graphicx,graphics}
\usepackage{url}
\usepackage{color}
\newtheorem{theorem}{Theorem}[section]

\newcommand{\SNR}{\mathsf{SNR}}

\newcommand{\eqFunc}{\overset{\mathrm{f}}{=}}
\newcommand{\C}{\mcal{C}}
\newcommand{\G}{\mathrm{G}}

\newcommand{\uplink}{\mathrm{ul}}
\newcommand{\downlink}{\mathrm{dl}}
\newcommand{\Out}{\mathrm{out}}
\newcommand{\Inn}{\mathrm{in}}

\newcommand{\lp}{\left(}
\newcommand{\rp}{\right)}

\newcommand{\lbp}{\left\{}
\newcommand{\rbp}{\right\}}
\newcommand{\lba}{\left\lvert}
\newcommand{\rba}{\right\rvert}

\newcommand{\mcal}{\mathcal}
\newcommand{\mscr}{\mathscr}

\newcommand{\mb}{\mathbf}
\newcommand{\mbb}{\mathbb}
\newcommand{\msf}{\mathsf}

\IEEEoverridecommandlockouts

\allowdisplaybreaks

\usepackage{cite,subfigure}
\title{On Two-Pair Two-Way Relay Channel with an Intermittently Available Relay}
\author{
\authorblockN{Shih-Chun Lin}
\authorblockA{
Department of ECE, NTUST\\
Taipei, Taiwan\\
\textsf{sclin@mail.ntust.edu.tw}} \and
\authorblockN{I-Hsiang Wang}
\authorblockA{
Department of EE, NTU\\
Taipei, Taiwan\\
\textsf{ihwang@ntu.edu.tw}}
\thanks{
The work of S.-C. Lin was supported by Ministry of Education,
Taiwan under grants "Aiming For the Top University Program" and
Ministry of Science and Technology, Taiwan, under Grants MOST
101-2221-E-011-170-MY3. The work of I.-H. Wang was supported by
Ministry of Science and Technology, Taiwan, under Grants MOST
103-2221-E-002-089-MY2 and MOST 103-2622-E-002-034.} }

\begin{document}
\maketitle
\begin{abstract}
When multiple users share the same resource for physical layer cooperation such as relay terminals in their vicinities, this shared resource may not be always available for every user, and it is critical for transmitting terminals to know whether other users have access to that common resource in order to better utilize it.
Failing to learn this critical piece of information may cause severe issues in the design of such cooperative systems.
In this paper, we address this problem by investigating a two-pair two-way relay channel with an intermittently available relay. In the model, each pair of users need to exchange their messages within their own pair via the shared relay. The shared relay, however, is only intermittently available for the users to access. The accessing activities of different pairs of users are governed by independent Bernoulli random processes.
Our main contribution is the characterization of the capacity region to within a bounded gap in a symmetric setting, for both delayed and instantaneous state information at transmitters. An interesting observation is that the bottleneck for information flow is the quality of state information (delayed or instantaneous) available at the relay, not those at the end users. To the best of our knowledge, our work is the first result regarding how the shared intermittent relay should cooperate with multiple pairs of users in such a two-way cooperative network.
\end{abstract}

%\vspace{-3.5mm}
\section{Introduction}
Physical layer cooperation has been proposed as a promising approach to increase spectral efficiency, where additional resources are dedicated for cooperation, such as relay terminals in the vicinity. Such resources for cooperation could be shared by many different users.
One of the envisioned scenarios for physical layer cooperation is
\emph{multi-pair two-way communication via a relay}, where multiple pairs
of users exchange their messages within their own pairs, with the
help of a relay. The shared resource for cooperation in this scenario is the relay shared by multiple pairs of users.
The simplest information theoretic model for studying this problem is the \emph{two-way relay channel} without user-to-user connections. There has been a great deal of works focusing on (multi-pair) two-way relay channels, such as \cite{ChungTwoWay,SezginAvestimehr_12}. A conventional assumption in these works is that, the relay is always available for the users to access, so that they can exchange data via the relay all the time.

%In multi-pair two-way communication via a relay, the key design question is how the relay
%broadcasts signals to all users to help deliver the respective
%messages, upon receiving the superposition of signals sent from
%multiple pairs of users.
%%See Figure~\ref{fig:Scenario} for an illustration.
%The simplest information theoretic model for studying this problem is the \emph{two-way relay channel} without user-to-user connections. There has been a great deal of works focusing on (multi-pair) two-way relay channels, such as \cite{ChungTwoWay,SezginAvestimehr_12}. A conventional assumption made in these works is that, the relay is always available for the users to access, so that they can exchange data via the relay all the time.

%\begin{figure}[htbp]
%\centering
%\includegraphics[width = \linewidth]{Scenario.pdf}
%\caption{Multi-Pair Two-Way Communication via a Shared Relay}
%\label{fig:Scenario}
%\end{figure}

In practice, however, the opportunity of cooperation may not always exist, mainly because the management and allocation of resources for cooperation (such as relay terminals in their vicinities) lies beyond the physical layer. When multiple users share the same cooperation resource,
%it is critical for transmitters to know whether or not other users have access to that common resource in order to better utilize it.
it may severely impact the design of such cooperative systems if transmitters cannot timely learn how heavily the common resource is currently being utiliized.
In the context of multi-pair two-way communication,
%each pair of users may not be able to access the relay all the time.
the issue becomes relevant especially when the spectral activity such
as the frequency hopping sequence and/or the frequency coding
pattern of a communication link is unknown to a relay which is
installed by a third party {\cite{METIS}} but shared by
multiple pairs of users. Hence, it is of fundamental interest to characterize the capacity of such systems, under various levels of state information availability of other pairs' accessing activities.

In this paper, we take a first step towards this direction by
investigating a two-pair two-way relay channel where the two
pairs get to access the relay intermittently, under various settings of temporal availability of \emph{activity state information} at transmitters.
The availability of
accessing the relay is governed by two independent Bernoulli $p$
i.i.d. processes, one for each pair. The terminals can either have
delayed information about the activity states, or instantaneous
state information.
See Figure~\ref{fig:Model} for an illustration of the channel model.

\begin{figure}[htbp]
\centering
\includegraphics[width = 0.65\linewidth]{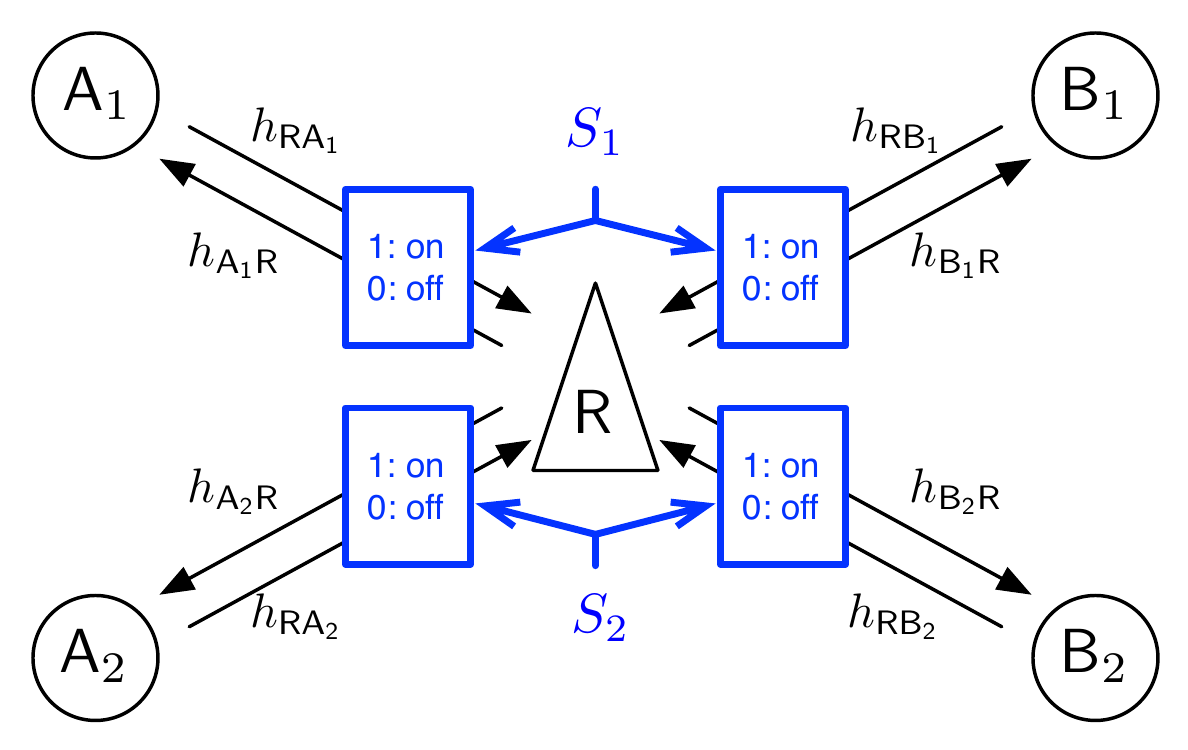}
\caption{Two-Pair Two-Way Relay Channel with an Intermittent Relay} \label{fig:Model}
\end{figure}

Our main contribution is the characterization of the capacity
region to within a bounded gap in a symmetric setup, both under
the delayed state information setting and the instantaneous state
information setting. We show that the two-pair two-way relay channel can be decomposed
into the uplink and the downlink part, and the approximate
capacity region is characterized as the intersection of the uplink
outer bound region and the downlink outer bound region. The
decomposition principle can be viewed as an extension of {that in}
the multi-pair two-way relay channel with a static relay
\cite{SezginAvestimehr_12}. An interesting observation is hence
that the bottleneck for information flow within the system is the
quality of state information available at the relay. Towards establishing the achievability of the bound-gap result,
for the downlink phase with delayed state information, we have developed a novel scheme that takes care of unequal received signal-to-noise ratios. The scheme complements that in \cite{MaddahAliAllerton13} where equal received SNRs are assumed.

We obtain key insights from the binary expansion model \cite{AvestimehrDiggavi_11} for this problem to develop our scheme, where the main novelty is two-fold. First, since the state is not known instantaneously at the relay (transmitter), a lattice-based dirty paper coding (DPC) is employed instead of conventional DPC based on Gaussian random codes. Second, to take care of the unequal received SNRs, instead of quantizing the erased sequences into a single codeword like \cite{MaddahAliAllerton13}, we propose a successive refinement framework so that stronger receiver can have higher resolution into the quantized signal.

{\bf Related work}:
Two-way relay channel with a static relay has been extensively studied.
%\cite{ChungTwoWay,SezginAvestimehr_12,AvestimehrDiggavi_11}.
For the single-pair two-way relay channel, \cite{ChungTwoWay} characterized the capacity region to within $\frac{1}{2}$ bit with compute-and-forward \cite{Nazer} and cut-set based outer bound.
\cite{SezginAvestimehr_12} extended the result to the two-pair two-way relay channel, using insights from the binary-expansion model \cite{AvestimehrDiggavi_11}. However, when the relay is
intermittently available, there has been very few results
regarding how the shared relay should cooperate with multiple
pairs of users. Related works that address intermittence in wireless networks were focused on bursty interference networks.
%Some related works were focused on one-hop bursty
%interference channel where the interference is intermittent
%\cite{WangISIT13,MaddahAliTIT15,JafarISIT13}.
\cite{WangISIT13} characterized the generalized degrees of freedom of a bursty interference channel with delayed state information and channel output feedback, while \cite{MaddahAliTIT15} \cite{JafarISIT13} studied the degrees of freedom of binary fading interference channels with instantaneous or delayed
state information. However, the intermittent availability of cooperation resources have not been investigated widely.
%It was shown that in \cite{MaddahAliTIT15}, even
%with only delayed state feedback, significant rate gain can be
%obtained compared with the system without feedback.

%The rest of this paper is organized as follows. In Section
%\ref{sec_Formulation}, we formulate our problem in the binary
%expansion and Gaussian models. Then in Sec. \ref{sec_Result}, we
%provide capacity region for the binary expansion model and use it
%to approximate that for the Gaussian model. Finally, numerical
%results are provided in Sec. \ref{Sec_numerical}.

%\vspace{-3.5mm}
\section{Problem Formulation}\label{sec_Formulation}
%\vspace{-1mm}
\subsection{Channel Model}
In the system, there are two pairs of \emph{end user} terminals, pair 1: $\lp
\msf{A}_1,\msf{B}_1\rp$ and pair 2: $\lp \msf{A}_2,\msf{B}_2\rp$,
and one \emph{relay} terminal $\msf{R}$. Each terminal can listen and
transmit simultaneously, and the blocklength is $N$. End user
$\msf{U}_i$ in pair $i$ ($\msf{U}=\msf{A},\msf{B}$, $i=1,2$) would
like to deliver its message $W_{\msf{U}_i}$ to the other end user in pair $i$. The encoding constraints depend on the state
information assumption and are detailed in
Section~\ref{subsec_SI}.

%\subsection{Channel Model}
The two-pair two-way Gaussian relay channel with an intermittent
relay is depicted in Figure~\ref{fig:Model} and defined as
follows. The transmitted signals of the five terminals are
$X_{\msf{A}_1}, X_{\msf{B}_1}, X_{\msf{A}_2}, X_{\msf{B}_2},
X_{\msf{R}} \in \mbb{C}$ respectively, each of which is subject to
unit power constraint, and the received signals are
\begin{align}
Y_{\msf{A}_i}[t] &= h_{\msf{A}_i\msf{R}} S_i[t] X_{\msf{R}}[t] + Z_{\msf{A}_i}[t],\ i=1,2, \notag \\
Y_{\msf{B}_i}[t] &= h_{\msf{B}_i\msf{R}} S_i[t] X_{\msf{R}}[t] + Z_{\msf{B}_i}[t],\ i=1,2, \label{eq_Gau_DL_channel}\\
%Y_{\msf{R}}[t] &= \sum_{i=1,2}h_{\msf{R}i} S_i[t] \lp X_{\msf{A}_i}[t]+X_{\msf{B}_i}[t]\rp + Z_{\msf{R}}[t],
Y_{\msf{R}}[t] &= \sum_{i=1,2}h_{\msf{R}\msf{A}_i} S_i[t]
X_{\msf{A}_i}[t] \!+\! h_{\msf{R}\msf{B}_i} S_i[t]X_{\msf{B}_i}[t]
\!+\! Z_{\msf{R}}[t], \label{eq_Gau_UP_channel}
\end{align}
where the independent additive noises at the five terminals
$Z_{\msf{A}_1}[t], Z_{\msf{B}_1}[t], Z_{\msf{A}_2}[t],
Z_{\msf{B}_2}[t], Z_{\msf{R}}[t]$ are $\mcal{CN}\lp 0,1\rp$ i.i.d.
over time. $\lbp S_i[t]\rbp$ denotes the random process that
governs the accessing activity of the two users in pair $i$, for
$i=1,2$. $\lbp S_1[t]\rbp$ and $\lbp S_2[t]\rbp$ are independent
Bernoulli $p$ processes, i.i.d. over time \footnote{In general,
the states may be correlated across time and thus allowing us to
predict future and improve the throughput. However, discussing the
benefit of predicting the future is beyond the scope of this
paper, and thus as \cite{WangISIT13}\cite{MaddahAliTIT15}, we
impose i.i.d assumptions on states.}. We denote the
signal-to-noise ratios as follows: for $i=1,2$,
\begin{align*}
\SNR_{\msf{R}\msf{A}_i} &:= \lba h_{\msf{R}\msf{A}_i}\rba^2 & \SNR_{\msf{R}\msf{B}_i} &:= \lba h_{\msf{R}\msf{B}_i}\rba^2\\
\SNR_{\msf{A}_i\msf{R}} &:= \lba h_{\msf{A}_i\msf{R}}\rba^2 &
\SNR_{\msf{B}_i\msf{R}} &:= \lba h_{\msf{B}_i\msf{R}}\rba^2
\end{align*}

{Note that we focus the fast fading scenario where a codeword can
span over different activity states. This assumption makes our
uplink model \eqref{eq_Gau_UP_channel} fundamentally different to
the random access channel in \cite{MineroTIT12}. In
\cite{MineroTIT12}, the slow fading scenario was studied where
encoding over different states was prohibited.}

\subsection{Activity State Information}\label{subsec_SI}
We consider two scenarios in this paper regarding how the
accessing activity state processes $\{S_1[t]\}$ and $\{ S_2[t]\}$
are known to the five terminals, in terms of how the state
information helps in encoding.
\subsubsection{Delayed State Information}
\begin{itemize}
\item For end users: for user $\msf{U}_i$ in pair $i$
($\msf{U}=\msf{A},\msf{B}$, $i=1,2$), $X_{\msf{U}_i}[t] \eqFunc
\lp W_{\msf{U}_i}, Y_{\msf{U}_i}^{t-1}, S_1^{t-1}, S_2^{t-1}\rp$.
%\lp W_{\msf{U}_i}, Y_{\msf{U}_i}^{t-1}, S_i^{t-1}, S_j^{t-1}\rp$,
%where $j = i \oplus 1$.
\item For the relay: $X_{\msf{R}}[t]
\eqFunc \lp Y_{\msf{R}}^{t-1}, S_1^{t-1}, S_2^{t-1}\rp$.
\end{itemize}

\subsubsection{Instantaneous State Information}
\begin{itemize}
\item For end users: for user $\msf{U}_i$ in pair $i$
($\msf{U}=\msf{A},\msf{B}$, $i=1,2$), $X_{\msf{U}_i}[t] \eqFunc
\lp W_{\msf{U}_i}, Y_{\msf{U}_i}^{t-1}, S_1^{t}, S_2^{t}\rp$.
%\lp W_{\msf{U}_i}, Y_{\msf{U}_i}^{t-1}, S_i^{N}, S_j^{t}\rp$,
%where $j = i \oplus 1$.
\item For the relay: $X_{\msf{R}}[t]
\eqFunc \lp Y_{\msf{R}}^{t-1}, S_1^{t}, S_2^{t}\rp$.
\end{itemize}

The capacity region $\mscr{C}$ depends on the available activity
state information. We take the following notation to denote the
capacity region under certain setting of activity state
information: $\mscr{C}\lp \mathrm{u, r}\rp$, where the first
argument $\mathrm{u \in \lbp d,i\rbp}$ denotes that the end users have delayed state information ($\mathrm{d}$) or
instantaneous state information ($\mathrm{i}$), while the second
argument $\mathrm{r \in \lbp d,i\rbp}$ denotes the type of the available activity state information at the relay terminal.
%{\red It should be noted that, for
%user pair $i$, the state information $S_j$ of the other user pair
%does not influence our results given in the next section.}

\section{Main Results}\label{sec_Result}
In this paper, we focus on the symmetric case where
$\SNR_{\msf{R}\msf{U}_i} = \SNR_{\msf{R}i}$,
$\SNR_{\msf{U}_i\msf{R}} = \SNR_{i\msf{R}}$,
%$n_{\msf{R}\msf{U}_i} = n_{\msf{R}i}$, $n_{\msf{U}_i\msf{R}} = n_{i\msf{R}}$,
for $\msf{U}=\msf{A},\msf{B}$ and $i=1,2$. We focus on characterizing the symmetric rate tuple $\lp R_1, R_2\rp$, where
$R_{\msf{A}_i} = R_{\msf{B}_i} = R_i$ for $i=1,2$.
Without loss of generality, we assume that $\SNR_{1\msf{R}} \geq \SNR_{2\msf{R}}$.
% and $n_{1\msf{R}} \geq n_{2\msf{R}}$.

To present our main result, let us begin with some definitions useful in characterizing the approximate capacity regions.
\par
{\it Notations}:
\begin{itemize}
\item
Define $\C(x):=\log(1+x)$ (logarithm is of base 2).
%\item For $a,b\in\mathbb{R}$, $a\vee b := \max\lbp a, b\rbp$, $a\wedge b := \min\lbp a,b\rbp$.
\item For a $\mscr{R} \subseteq \mbb{R}^2$, define the pointwise minus operator $\ominus$ as follows: $\mscr{R} \ominus (a,b) := \lbp \lp x-a,y-b\rp : \lp x,y\rp\in\mscr{R}\rbp$.
\end{itemize}
%\begin{definition}[Uplink Rate Regions] \label{Def_up}
%\begin{itemize}
%\item Binary expansion model: let $\mscr{C}^{\uplink}_{\BE}\lp
%\mathrm{d}\rp = \mscr{C}^{\uplink}_{\BE}\lp \mathrm{i}\rp$ be the
%collection of $\lp R_1,R_2\rp \ge 0$ satisfying
%\begin{align}
%\textstyle\frac{R_{1}}{p} &\leq n_{\msf{R}1},\quad
%\textstyle\frac{R_{2}}{p} \leq n_{\msf{R}2}, \label{eq_bin_up_full_C1}\\
%\textstyle\frac{R_{1}}{p}+\frac{R_{2}}{p}&\leq
%(1-p)(n_{\msf{R}1}+n_{\msf{R}2})+p \lp n_{\msf{R}1}\vee
%n_{\msf{R}2}\rp.\label{eq_bin_up_full_C3}
%\end{align}
\par\smallskip
%\item Gaussian model:
{\noindent \bf Uplink Rate Regions}:
Let $\mscr{R}^{\uplink}_{ \Out}\lp
\mathrm{d}\rp$ be the collection of $\lp R_1,R_2\rp \ge 0$
satisfying
\begin{align}
\textstyle\frac{R_{1}}{p} &\leq \C\lp \SNR_{\msf{R}1}\rp,\quad
\textstyle\frac{R_{2}}{p} \leq \C\lp \SNR_{\msf{R}2}\rp, \label{eq_up_delayedub_C1} \\
\textstyle\frac{R_{1}}{p}+\frac{R_{2}}{p} &\leq
(1-p)\lp \C\lp\SNR_{\msf{R}1}\rp+ \C\lp\SNR_{\msf{R}2}\rp\rp \notag \\
&\hspace{-36pt} + p\,\C\lp\SNR_{\msf{R}1}+\SNR_{\msf{R}2}+2\sqrt{\SNR_{\msf{R}1}\SNR_{\msf{R}2}}\rp.
\label{eq_up_delayedub_C3}
\end{align}
Let $\mscr{R}^{\uplink}_{ \Inn}\lp \mathrm{d}\rp :=
\mscr{R}^{\uplink}_{ \Out}\lp \mathrm{d}\rp \ominus \lp 1,1\rp$.
Let $\mscr{R}^{\uplink}_{ \Out}\lp \mathrm{i}\rp$ be the
collection of $\lp R_1,R_2\rp \ge 0$ satisfying
\eqref{eq_up_delayedub_C1} -- \eqref{eq_up_delayedub_C3} with
$\SNR$'s replaced by $\frac{\SNR}{p}$, and $\mscr{R}^{\uplink}_{
\Inn}\lp \mathrm{i}\rp$ be $\mscr{R}^{\uplink}_{ \Inn}\lp
\mathrm{d}\rp$  with $\SNR$'s replaced by $\frac{\SNR}{p}$.
%\end{itemize}
%\end{definition}

%\begin{definition}[Downlink Rate Regions] \label{Def_downlink}
%\begin{itemize}
%\item Binary expansion model: let $\mscr{C}^{\downlink}_{\BE}\lp
%\mathrm{d}\rp$ be the collection of $\lp R_1,R_2\rp \ge 0$
%satisfying
%\begin{align}
%%\textstyle\frac{R_{1}}{p} &\leq n_{\msf{R}1},\quad
%\textstyle\frac{R_{2}}{p} &\leq n_{2\msf{R}}, \label{eq_BE_BC_UB_R2_1} \\ \textstyle\frac{R_{1}}{p}+\frac{R_{2}}{p(2-p)} &\leq n_{1\msf{R}}, \label{eq_BE_BC_UB_R1}\\
%\textstyle\frac{R_{1}}{p(2-p)}+\frac{R_{2}}{p}&\leq
%\textstyle\frac{1}{(2-p)}(n_{1\msf{R}}-n_{2\msf{R}})+n_{2\msf{R}}.
%\label{eq_BE_BC_UB_R2}
%\end{align}
%Let $\mscr{C}^{\downlink}_{\BE}\lp \mathrm{i}\rp$ be the
%collection of $\lp R_1,R_2\rp \ge 0$ satisfying
%\eqref{eq_bin_up_full_C1} -- \eqref{eq_bin_up_full_C3}, with
%$n_{\msf{R}1}$ and $n_{\msf{R}2}$ replaced by $n_{1\msf{R}}$ and
%$n_{2\msf{R}}$ respectively.
\par\smallskip
%\item Gaussian model:
{\noindent \bf Downlink Rate Regions}:
Let $\mscr{R}^{\downlink}_{ \Out}\lp
\mathrm{d}\rp$ be the collection of $\lp R_1,R_2\rp \ge 0$
satisfying
\begin{align}
\textstyle\frac{R_{2}}{p} &\leq \C\lp \SNR_{2\msf{R}}\rp, \label{eq_AWGN_BC_UB_R2_1} \\
\textstyle\frac{R_{1}}{p}+\frac{R_{2}}{p(2-p)} & \leq \C\lp \SNR_{1\msf{R}}\rp, \label{eq_AWGN_BC_UB_R1}\\
\textstyle\frac{R_{1}}{p(2-p)}+\frac{R_{2}}{p} &\leq \textstyle\frac{\C\lp \SNR_{1\msf{R}}\rp - \C\lp \SNR_{2\msf{R}}\rp}{2-p} + \C\lp \SNR_{2\msf{R}}\rp \label{eq_AWGN_BC_UB_R2}.
\end{align}
Let $\mscr{R}^{\downlink}_{ \Inn}\lp \mathrm{d}\rp :=
\mscr{R}^{\downlink}_{ \Out}\lp \mathrm{d}\rp \ominus \lp \Delta_1,\Delta_2\rp$, where
\begin{align}
\Delta_1&= \textstyle\frac{p(1-p)}{3-p}\log3+\frac{p}{3-p}\log\frac{2\pi e}{12}, \label{eq_Distortion1}\\
\Delta_2&= \textstyle\max\left\{p, \frac{p(1-p)}{3-p}\log10+\frac{p}{3-p}\right\}. \label{eq_Distortion2}
\end{align}
Let
$\mscr{R}^{\downlink}_{ \Out}\lp \mathrm{i}\rp = \mscr{R}^{\downlink}_{ \Inn}\lp \mathrm{i}\rp$ be the
collection of $\lp R_1,R_2\rp \ge 0$ with
\begin{align*}
\textstyle\frac{R_{1}}{p} &\textstyle \leq \C\lp
\frac{\SNR_{1\msf{R}}}{p(2-p)}\rp,\quad
\textstyle\frac{R_{2}}{p} \textstyle \leq \C\lp \frac{\SNR_{2\msf{R}}}{p(2-p)}\rp,\\
\textstyle\frac{R_{1}}{p}+\frac{R_{2}}{p} &\leq \textstyle(1-p)\lp
\C\lp \frac{\SNR_{1\msf{R}}}{p(2-p)}\rp+ \C\lp
\frac{\SNR_{2\msf{R}}}{p(2-p)}\rp\rp
\\ &\quad \textstyle + p\, \C\lp
\frac{\SNR_{1\msf{R}}+\SNR_{2\msf{R}}}{p(2-p)}\rp.
\end{align*}
%\end{itemize}
%\end{definition}

Before we proceed, we provide some numerical evaluations to
illustrate various regions defined above. We set the on/off
probability of activity state $p=0.6$. In Figure \ref{fig:Gup}, we
show $\mscr{R}^{\uplink}_{\Inn}\lp\mathrm{d}\rp$ and
$\mscr{R}^{\uplink}_{\Inn}\lp\mathrm{i}\rp$ with
$\SNR_{\msf{R}1}=30+20\log1.5$ dB and $\SNR_{\msf{R}2}=30$ dB. In
Figure \ref{fig:GDL}, we show
$\mscr{R}^{\downlink}_{\Inn}\lp\mathrm{d}\rp$ and
$\mscr{R}^{\downlink}_{\Inn}\lp\mathrm{i}\rp$ with
$\SNR_{1\msf{R}}=30+20\log1.5$ dB and $\SNR_{2\msf{R}}=30$ dB.
%Note that $\mscr{R}^{\downlink}_{\Inn}\lp\mathrm{d}\rp$ is the smallest
%among four regions in Figure \ref{fig:Gup} and \ref{fig:GDL}. If
%the relay has only delayed state information, the downlink from
%relay becomes the bottleneck for information flow within the system.

{\bf Remark}: Note that
$\mscr{R}^{\downlink}_{\Inn}\lp\mathrm{d}\rp$ is the smallest
among four regions in Figure \ref{fig:Gup} and \ref{fig:GDL}. If
the relay has only delayed state information, the downlink from
relay becomes the bottleneck for information flow within the
system.

\begin{figure}[htbp]
\centering{ \subfigure[]{\includegraphics [width
=0.6\linewidth]{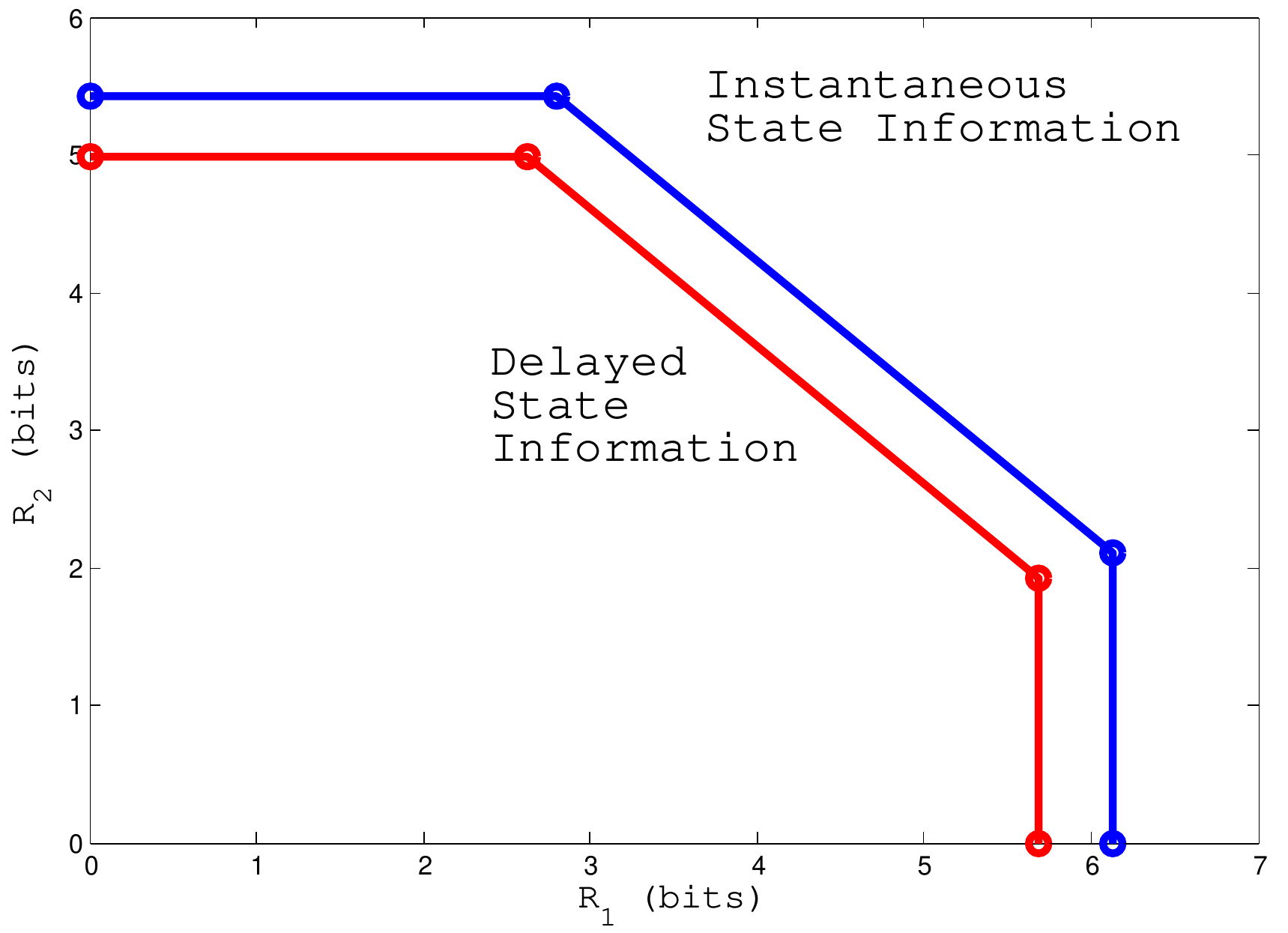}\label{fig:Gup}}
\subfigure[]{\includegraphics [width
=0.6\linewidth]{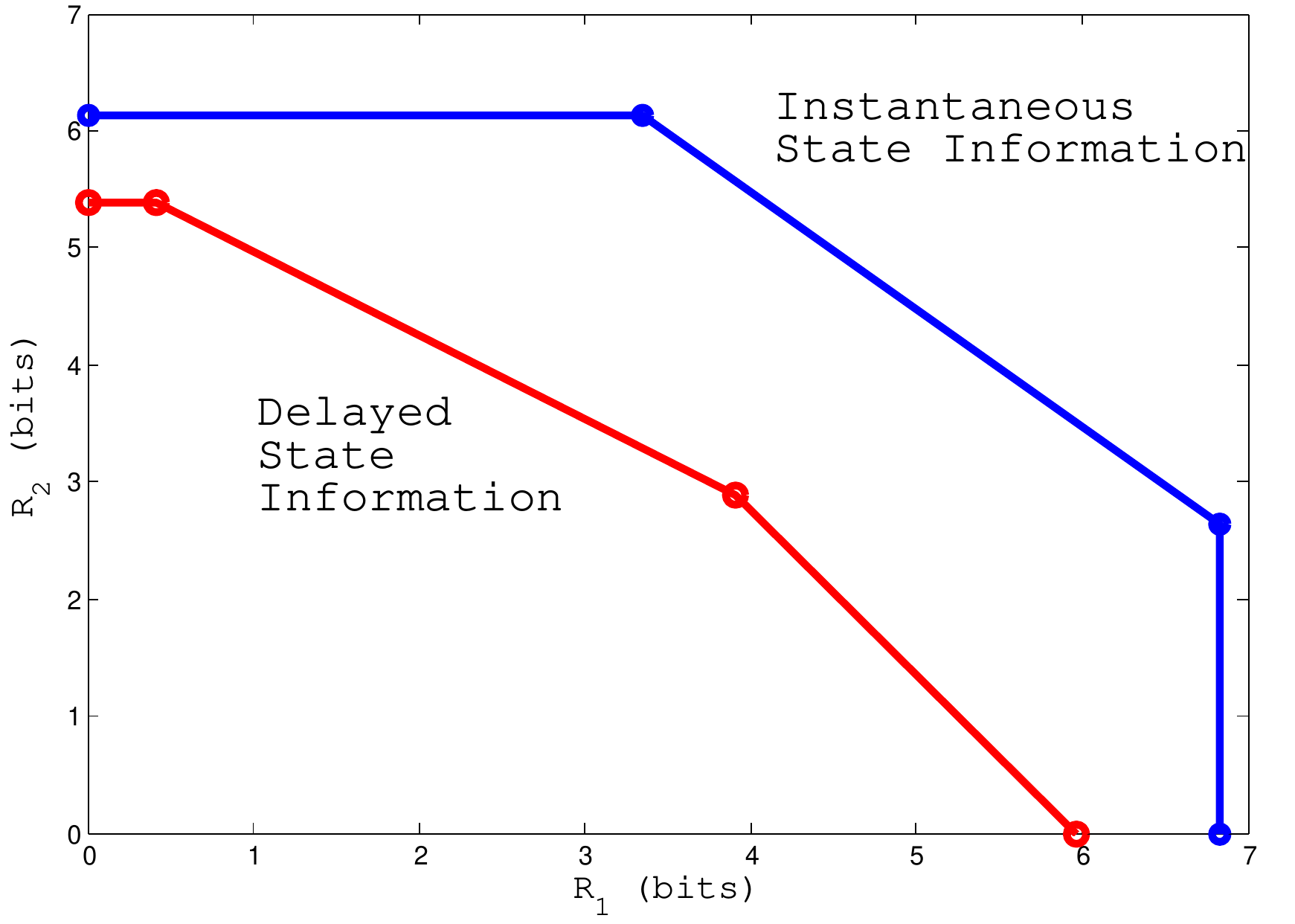}\label{fig:GDL}}} \caption{Bounded-gap
uplink (a) and downlink (b) inner bound regions with delayed and
instantaneous activity state information.}
\end{figure}

%\begin{figure}[htbp]
%\centering
%\includegraphics[width = 0.6\linewidth]{Guppaper}
%\caption{Bounded-gap uplink inner bound regions with delayed and
%instantaneous activity state information.} \label{fig:Gup}
%\end{figure}
%
%\begin{figure}[htbp]
%\centering
%\includegraphics[width = 0.6\linewidth]{GDLpaper}
%\caption{Bounded-gap downlink inner bound regions with delayed and
%instantaneous activity state information.} \label{fig:GDL}
%\end{figure}

Our main result is summarized in the following theorem.
\begin{theorem}[Capacity Region to within a Bounded Gap]\label{Theo_Main}
For capacity region $\mscr{C}\lp\mathrm{u,r}\rp$, we have inner
and outer bounds
\begin{align}
\mscr{C}\lp\mathrm{u,r}\rp &\supseteq
\mscr{R}^{\uplink}_{\Inn}\lp\mathrm{u}\rp \cap \mscr{R}^{\downlink}_{\Inn}\lp\mathrm{r}\rp,\ \forall\, \lp\mathrm{u,r}\rp \in \lbp \mathrm{d,i}\rbp^2, \label{eq_Gau_CapIn} \\
\mscr{C}\lp\mathrm{u,r}\rp &\subseteq
\mscr{R}^{\uplink}_{\Out}\lp\mathrm{u}\rp \cap
\mscr{R}^{\downlink}_{\Out}\lp\mathrm{r}\rp,\ \!\!\forall\,
\lp\mathrm{u,r}\rp \in \lbp \mathrm{d,i}\rbp^2\!. \!
\label{eq_Gau_CapOuter}
\end{align}
Since for all $\lp\mathrm{u,r}\rp \in \lbp \mathrm{d,i}\rbp^2$, $\mscr{R}^{\uplink}_{\Out}\lp\mathrm{u}\rp$ and $\mscr{R}^{\uplink}_{\Inn}\lp\mathrm{u}\rp$ are within a bounded gap, and so are $\mscr{R}^{\downlink}_{\Out}\lp\mathrm{r}\rp$ and $\mscr{R}^{\downlink}_{\Inn}\lp\mathrm{r}\rp$, we have characterized the capacity region to within a bounded gap.
\end{theorem}

\begin{proof}
Regarding the proof of the converse, we employ cut-set based outer
bounds and enhance the downlink channel to a degraded broadcast
channel where feedback does not increase the capacity region
\cite{GeorgiadisDetDelayedBC,BooK_NIT_KIM}. Details can be found
in the Appendix.

Regarding the achievability, here we provide the scheme for the inner bound of
$\mscr{C}\lp\mathrm{d,d} \rp$ in \eqref{eq_Gau_CapIn}, the case
where end users and relay all have delayed state information. The proofs for the other three
combinations in Theorem \ref{Theo_Main} easily follow, and are
also provided in the Appendix.

Our scheme consists of two phases: the uplink phase and the downlink
phase. In the uplink phase, the relay terminal aims to decode the
two XORs of the two pairs of messages $\Sigma_i=W_{\msf{A}_i}
\oplus W_{\msf{B}_i},\;i=1,2$ from its received signal, and store
them for later uses. Hence, it can be viewed as a function
computation problem over a multiple access channel. In the
downlink phase, the relay terminal re-encodes the stored XORs
$\lbp \Sigma_1,\Sigma_2\rbp$ and delivers $\Sigma_i$ to end users
$\lbp \msf{A}_i,\msf{B}_i\rbp$ for $i=1,2$. The end user terminals
decode their desired messages from the XORs by using its self
message as side information. Hence, it can be viewed as a
broadcast channel with two independent messages $\lbp
\Sigma_1,\Sigma_2\rbp$ and four receivers
$\lbp\msf{A}_1,\msf{B}_1,\msf{A}_2,\msf{B}_2\rbp$, where $\lbp
\msf{A}_i, \msf{B}_i\rbp$ aim to decode $\Sigma_i$, for $i=1,2$.

Further note that in the symmetric setting, since the rate of the
messages $W_{\msf{A}_i}$ and $W_{\msf{B}_i}$ are both $R_i$, the
rate of the XOR $\Sigma_i$ is also $R_i$, for $i=1,2$. Hence, we
are able to establish the inner bound region of achievable $\lp
R_1, R_2\rp$ as the intersection of the inner bound region of the
uplink phase and that of the downlink phase, denoted by
$\mscr{R}^{\uplink}_{\Inn}$ and $\mscr{R}^{\downlink}_{\Inn}$
respectively. Below we give the proof sketches for the uplink
phase in Sec. \ref{sec:PfInner} and the downlink phase in Sec.
\ref{sec:PfInnerdL}. The detailed proofs %as well as the outer bounds
are given in the Appendix.
%The proofs for the other three
%combinations in Theorem \ref{Theo_Main} easily follow, and are
%also provided in the Appendix.
\end{proof}

\section{Proof Sketch of the Inner Bound $\mscr{R}^{\uplink}_{\Inn}\lp\mathrm{d}\rp$  in
\eqref{eq_Gau_CapIn} for Uplink with Delayed State
Information}\label{sec:PfInner}

To achieve $\mscr{R}^{\uplink}_{\Inn}\lp\mathrm{d}\rp$ in the
uplink phase, we will use lattice-based compute-and-forward
\cite{Nazer}. Casting it as a function computation problem over
the four-transmitter multiple access channel, the relay can
successfully decode the XORs of messages $ \Sigma_i=W_{\msf{A}_i}
\oplus W_{\msf{B}_i},\;i=1,2$ from its received signal
\eqref{eq_Gau_UP_channel} without explicitly decoding the four
messages $\lbp W_{\msf{A}_1}, W_{\msf{B}_1}, W_{\msf{A}_2},
W_{\msf{B}_2}\rbp$, thanks to the linearity of lattice codes.
Compared with the scheme in \cite{SezginAvestimehr_12}, our scheme
needs to deal with the additional ergodic activity states
$\{S_1(t),S_2(t)\}$ and the delayed state information. Also we
adopt joint lattice deocding from \cite{scTWCOM14}, which has
better performance than the successive lattice decoding in
\cite{SezginAvestimehr_12}.

The details of achieving
$\mscr{R}^{\uplink}_{\Inn}\lp\mathrm{d}\rp$ in
\eqref{eq_Gau_CapIn} come as follows. First, we assume that the
channel gains in \eqref{eq_Gau_UP_channel} are real, which is
without loss of generality since we can pre-rotate the phase of
the complex channel before transmission. Then we collect the real
and imaginary parts of the $T$ received symbols at the relay as
\cite{scTWCOM14}, and focus on the following real equivalent
uplink channel from \eqref{eq_Gau_UP_channel} as
\begin{equation}\label{eq_ul_Lattice_model}
    \mathbf{y}_{\msf{R}}=\left[\mathbf{H}_{A_1} \; \mathbf{H}_{A_2}\right]\left[\begin{array}{c}
      \mathbf{x}_{\msf{A}_1}+\mathbf{x}_{\msf{B}_1}\\
      \mathbf{x}_{\msf{A}_2}+\mathbf{x}_{\msf{B}_2}\\
    \end{array}\right]+\mathbf{z}_{\msf
    R},
\end{equation}
where the $2T \times 1$ real vector $\mathbf{y}_{\msf{R}}$ is
formed from $Y_\msf{R}[t]$ as
\[\mathbf{y}_{\msf{R}}=\left[\mathrm{Re}(Y_\msf{R}[1]),\mathrm{Im}(Y_\msf{R}[1]),\ldots,
\mathrm{Re}(Y_\msf{R}[T]),\mathrm{Im}(Y_\msf{R}[T])\right]^T,\]
and $2T \times 1$
$\mathbf{x}_{\msf{A}_1},\mathbf{x}_{\msf{B}_1},\mathbf{x}_{\msf{A}_2},\mathbf{x}_{\msf{B}_2},\mathbf{z}_{\msf
R}$ are similarly formed from
$X_{\msf{A}_1}[t],X_{\msf{B}_1}[t],X_{\msf{A}_2}[t],X_{\msf{B}_2}[t],Z_{\msf
R}[t]$ respectively.  The $2T \times 2T$ diagonal channel matrix
for pair $i$ is
\begin{equation} \label{eq_Gau_up_HA}
\mathbf{H}_{\msf{A}_i}=|h_{\msf{RA}_i}|\cdot\mathrm{diag}(S_i(1),S_i(1),\ldots,S_i(T),S_i(T)).
\end{equation}
The transmitted vector for user $\msf{U}_i$ in pair $i$
($\msf{U}=\msf{A,B},i=1,2$) is
\begin{equation} \label{eq_ul_Lattice}
    \mathbf{x}_{\msf{U}_i}= \left([\mathbf{c_{\msf{U}_i}}-\mathbf{d}_{\msf{U}_i}]\right) \; \mathrm{mod}  \;
    \Lambda_S.
\end{equation}
With  $\Lambda_i$ being the coding lattice
\cite{scTWCOM14}\cite{nested_lattice}, the message $W_{\msf{U}_i}$
is encoded using lattice codeword $\mathbf{c_{\msf{U}_i}} \in
\Lambda_i$, and the shaping lattice $\Lambda_S \subset \Lambda_i$.
As \cite{scTWCOM14}\cite{nested_lattice}, the independent dither
$\mathbf{d}_{\msf{U}_i}$ is uniformly distributed in the Voronoi
region of the shaping lattice $\Lambda_S$, and $ \mathrm{mod}
\;\Lambda_S$ is the modulo-lattice operation. At the relay, it
performs joint lattice decoding for XORs $\Sigma_1$ and $\Sigma_2$
on the following post-processed received signal
\[
\left( \mathbf{Wy}_{\msf{R}}+
[(\mathbf{d}_{\msf{A}_1}+\mathbf{d}_{\msf{B}_1})^T \;\;
(\mathbf{d}_{\msf{A}_2}+\mathbf{d}_{\msf{B}_2})^T]^T \right) \;
\mathrm{mod} \; (\Lambda_S \times \Lambda_S).
\]
From \eqref{eq_ul_Lattice_model}, by choosing
$\mathbf{W}=2\mathbf{H}^T(2\mathbf{H}\mathbf{H}^T+\mathbf{I})^{-1}$
where $\mathbf{H}=\left[\mathbf{H}_{\msf{A}_1} \;
\mathbf{H}_{\msf{A}_2}\right]$, the achievable sum rate
$R_1/p+R_2/p$ has gap $2/p$ to the RHS of
\eqref{eq_up_delayedub_C3}. The other two rate constraints for
$\mscr{R}^{\uplink}_{\Inn}\lp\mathrm{d}\rp$ can be similarly
proved to be achievable.

\section{Proof Sketch of the Inner Bound $\mscr{R}^{\downlink}_{\Inn}\lp\mathrm{d}\rp$  in
\eqref{eq_Gau_CapIn} for Downlink with Delayed State Information}
\label{sec:PfInnerdL}

In our symmetric setting, since $\msf{A}_i$ and $\msf{B}_i$ have
the same receiver SNRs and are under the same activity state $\lbp
S_i[t]\rbp$, for $i=1,2$, we can treat the downlink
%$\mscr{R}^{\downlink}_{\Inn}\lp\mathrm{d}\rp$ can be treated
as a broadcast channel \eqref{eq_Gau_DL_channel} where the relay
sends $\Sigma_1$ to user $\msf{B}_1$ and $\Sigma_2$ to user
$\msf{B}_2$ respectively, with delayed state information. Compared
with \cite{MaddahAliAllerton13}, which is focused on ergodic
Rayleigh fading downlink with equal received SNRs, our downlink
\eqref{eq_Gau_DL_channel} has different on/off channel statistics
and non-equal $\msf{SNR}_{1\msf{R}} \geq \msf{SNR}_{2\msf{R}}$.
These two differences raise new challenges for obtaining
bounded-gap capacity results.

%The bounded-gap achievement of the corner point
For the corner point of the outer bound region where
\eqref{eq_AWGN_BC_UB_R2} and \eqref{eq_AWGN_BC_UB_R2_1} intersect,
achieving it to within a bounded gap can be simply done by
Gaussian superposition coding. Thus we focus on the other corner
point where \eqref{eq_AWGN_BC_UB_R1} and \eqref{eq_AWGN_BC_UB_R2}
intersect:
\begin{align}
R_1 &=\textstyle
p\left(\C(\msf{SNR}_{1\msf{R}})\!-\!\C(\msf{SNR}_{2\msf{R}})\right)\!+\!
\frac{p(2-p)}{3-p}\C(\msf{SNR}_{2\msf{R}}), \label{eq_AWGN_BC_UB_point_2_R1} \\
R_2 &=\textstyle\frac{p(2-p)}{3-p}\C(\msf{SNR}_{2\msf{R}}).
\label{eq_AWGN_BC_UB_point_2_R2}
\end{align}
Our scheme to achieve $(R_1-\Delta_1,R_2-\Delta_2)$ with $R_1,
R_2$ taken from
\eqref{eq_AWGN_BC_UB_point_2_R1}\eqref{eq_AWGN_BC_UB_point_2_R2}
is a non-trivial extension of the scheme in
%the three-phase retransmission for
the binary erasure broadcast channel
\cite{GeorgiadisDetDelayedBC}, where $(\Delta_1,\Delta_2)$ are
given in \eqref{eq_Distortion1}\eqref{eq_Distortion2}. To obtain
insights, we start with a binary-expansion model
\cite{AvestimehrDiggavi_11} for this problem as follows.

\vspace{-1mm}
\subsection{Insights from Binary-Expansion Model} \label{subsec_BE_Model}
\vspace{-1.5mm}
\begin{figure}[htbp]
\centering
\includegraphics[width = \linewidth]{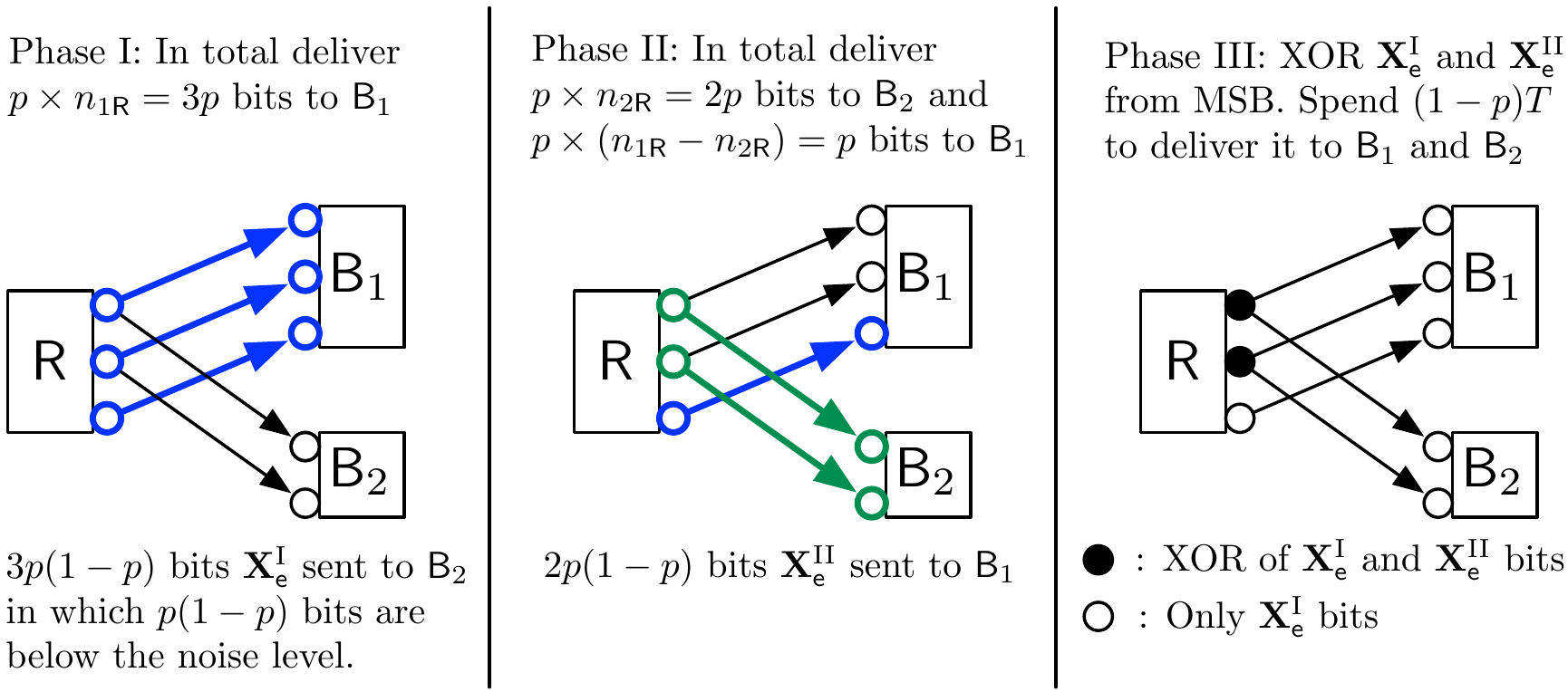}
\caption{Example for achieving corner point
\eqref{eq_BE_BC_UB_point_2} with
$(n_{1\msf{R}},n_{2\msf{R}})=(3,2)$ for the binary expansion
downlink with delayed state information} \label{fig:BE_downlink}
\end{figure}

In this subsection, we employ a binary expansion model
 corresponding to the downlink phase
\eqref{eq_Gau_DL_channel} to obtain insights. In this model, the
transmitted and received signals are binary vectors in
$\mbb{F}_2^{q}$, where $\mbb{F}_2$ denotes the binary field
$\{0,1\}$. The received signals are $ Y_{\msf{B}_i}[t] =
\mb{H}_{\msf{B}_i\msf{R}} S_i[t] X_{\msf{R}}[t],\ i=1,2, $ where
additions are modulo-two component-wise. Channel transfer matrices
are defined as follows: for $i=1,2$, $ \mb{H}_{\msf{B}_i\msf{R}}
:= \mb{S}^{q-n_{i\msf{R}}} $ where $q = \max_{i=1,2}\lbp
n_{i\msf{R}}\rbp$ and $\mb{S}\in\mathbb{F}_2^{q\times q}$ is the
shift matrix defined in \cite{AvestimehrDiggavi_11}. The corner
point corresponds to \eqref{eq_AWGN_BC_UB_point_2_R1} and
\eqref{eq_AWGN_BC_UB_point_2_R2} in this model is
\begin{equation} \label{eq_BE_BC_UB_point_2}
\textstyle (R_1,R_2)=\left(p(n_{1\msf{R}}-n_{2\msf{R}})+
\frac{p(2-p)}{3-p}n_{2\msf{R}}, \frac{p(2-p)}{3-p}
n_{2\msf{R}}\right).
\end{equation}

To achieve this point, the relay uses a three-phase coding scheme
extending that in \cite{GeorgiadisDetDelayedBC}. In Phases I and
II (each with block length $T$), the relay sends bits intended for
$\msf{B}_1$ and $\msf{B}_2$, using the top $n_{1\msf{R}}$ and
$n_{2\msf{R}}$ levels respectively. In addition, in Phase II the
relay also uses the bottom $\lp n_{1\msf{R}}-n_{2\msf{R}}\rp$
levels to deliver additional bits to $\msf{B}_1$. Hence,
$\msf{B}_1$ and $\msf{B}_2$ receive roughly $Tp\lp n_{1\msf{R}} +
n_{1\msf{R}}-n_{2\msf{R}}\rp$ and $Tpn_{2\msf{R}}$ desired bits in
Phase I and II respectively.

In Phase I, there will be roughly $Tp(1-p)n_{1\msf{R}}$ bits which
are erased at $\msf{B}_1$ but erroneously sent to $\msf{B}_2$ can
be used as side-information. We denote this length-$Tp(1-p)$
sequence of $n_{1\msf{R}}$-level binary vector by
$\mb{X}_{\msf{e}}^{\mathrm{I}}$. Note that the bottom $\lp
n_{1\msf{R}}-n_{2\msf{R}}\rp$ levels will lie \emph{below} the
noise level at $\msf{B}_2$ and will NOT appear in this binary
expansion model. Similarly in Phase II, there will be such a
length $Tp(1-p)$ sequence of $n_{2\msf{R}}$-level binary vector
intended for $\msf{B}_2$ but only received by $\msf{B}_1$. We
denote it by $\mb{X}_{\msf{e}}^{\mathrm{II}}$. We aim to
\emph{recycle} these bits in Phase III.

The block length of Phase III is roughly $Tp(1-p)$. In Phase III,
the relay makes use of delayed state information to form
$\mb{X}_{\msf{e}}^{\mathrm{I}}$ and
$\mb{X}_{\msf{e}}^{\mathrm{II}}$. Then it sends out
$\mb{X}_{\msf{e}}^{\mathrm{I}}\oplus
\mb{X}_{\msf{e}}^{\mathrm{II}}$ from the MSB level as depicted on
the rightmost of Figure~\ref{fig:BE_downlink}. Hence the bottom
$\lp n_{1\msf{R}}-n_{2\msf{R}}\rp$ levels consists of bits in the
bottom levels of $\mb{X}_{\msf{e}}^{\mathrm{I}}$ only. With side
information received in Phase I and II, each receiver can decode
the desired bits from the received XORs. In total the numbers of
bits \emph{recycled} in this phase are $Tp(1-p)n_{1\msf{R}}$ and
$Tp(1-p)n_{2\msf{R}}$ at $\msf{B}_1$ and $\msf{B}_2$ respectively.

Putting everything together, we achieve $R_2 =
\frac{p(2-p)}{3-p}n_{2\msf{R}}$ and $R_1 =
\frac{p\lbp(3-p)n_{1\msf{R}}-n_{2\msf{R}}\rbp}{3-p} =
p(n_{1\msf{R}}-n_{2\msf{R}})+\frac{p(2-p)}{3-p}n_{2\msf{R}}$.

%To achieve this point, the relay uses different coding schemes on
%the $n_{2\msf{R}}$ common layers connected to both $\msf{B}_1$ and
%$\msf{B}_2$ and the rest $n_{1\msf{R}}-n_{2\msf{R}}$ layers
%connected only to $\msf{B}_1$, as depicted in
%Fig.~\ref{fig:BE_downlink}. In the first two common layers in
%Fig.~\ref{fig:BE_downlink}, the relay broadcasts bits for
%$\msf{B}_1$ and $\msf{B}_2$ using the three-phase scheme in
%\cite{GeorgiadisDetDelayedBC}. In Phases I and II, the relay
%respectively sends binary bits intended for $\msf{B}_1$ and
%$\msf{B}_2$ only . The $n_{2\msf{R}}p(1-p)$ bits which are erased
%at $\msf{B}_1$ but erroneously sent to $\msf{B}_2$ can be used as
%side-information, and vice versa. With the delayed state feedback,
%in Phase III, the relay knows these erased bits and retransmits
%the XORs of them. With side-information, each receiver can decode
%the bits wanted from the received XORs. From the common layers,
%both $\msf{B}_1$ and $\msf{B}_2$ can decode
%$n_{2\msf{R}}(p+p(1-p))/(1+1+(1-p))$ bits correctly. In the rest
%$n_{1\msf{R}}-n_{2\msf{R}}$ layers (the third layer in
%Fig.~\ref{fig:BE_downlink}), the relay send additional independent
%$p(n_{1\msf{R}}-n_{2\msf{R}})$ bits to $\msf{B}_1$ during the
%three phases. Then point \eqref{eq_BE_BC_UB_point_2} is
%achievable.

\subsection{Proof Sketch of Bounded-gap Achievement to the Corner Point \eqref{eq_AWGN_BC_UB_point_2_R1} \eqref{eq_AWGN_BC_UB_point_2_R2} of the Outer Bound Region}
Extending
%achievement of corner point \eqref{eq_BE_BC_UB_point_2}
%for binary expansion model
to the Gaussian case,
%\eqref{eq_AWGN_BC_UB_point_2_R1} \eqref{eq_AWGN_BC_UB_point_2_R2},
we face the following two challenges. First, in Gaussian channel,
we are sending complex symbols instead of binary bits and there
will be additive Gaussian noise. Second, we need to incorporate
superposition coding into Phase II of Fig.~\ref{fig:BE_downlink},
while $\msf{B}_1$ may not be able to decode and cancel the
higher-layer codeword since the erasure state process at
$\msf{B}_1$ and $\msf{B}_2$ are different. Note that only signals
of $\msf{B}_2$ in Phase II have recycling from Phase III.

We solve the first challenge by resending erased symbols instead
of bits in the third phase. To do this, the relay will
\emph{quantize} the sum sequence formed by the erased symbols,
$X_{\msf{e}}^{\mathrm{I}}+ X_{\msf{e}}^{\mathrm{II}}$, and then
send out the quantization indices. Based on the insight learned in
the binary expansion model, we know that the resolution of
reconstruction must be different: $\msf{B}_1$ requires higher
resolution than $\msf{B}_2$ since $X_{\msf{e}}^{\mathrm{I}}$ goes
deeper in the bit levels. Hence, instead of directly quantizing
into a single quantization index, we employ \emph{successive
refinement} source coding \cite{cheng2005successive} so that
$\msf{B}_1$ is able to get a higher resolution in reconstruction.
Again gaining insights from the binary expansion model, since the
number of layers used by $\msf{B}_i$ is $n_{i\msf{R}}$ for $i=1,2$
in Phase I and II respectively, the MSE of the reconstruction at
$\msf{B}_i$ should be inverse proportional to
$\msf{SNR}_{i\msf{R}}, i=1,2$.

For the second challenge, we aim to solve it using dirty paper
coding (DPC).
%in Phase II, we must broadcast $\Sigma_1$ and
%$\Sigma_1$ simultaneously.
%it is well known that if instantaneous
%state information is available, the optimal scheme is to
%superimpose Gaussian signal for $\Sigma_2$ with dirty paper coding
%(DPC) for $\Sigma_1$.
However, the conventional DPC requires fully known channel
information $S_1(t)h_{\msf{B}_1\msf{R}}$ at the transmitter
\cite{nested_lattice}. In our case, the current on/off state
$S_1(t)$ is unknown at the relay.
%Thus
%instead of using multi-dimensional (codeword-based) lattice in
%\cite{nested_lattice},
Hence, we propose new one-dimensional (symbol-based) lattice
strategy to solve this problem.
%Moreover,
%in Phase III, the MSE of the reconstruction at $\msf{B}_i$ should
%be inverse proportional to $\msf{SNR}_{i\msf{R}}, i=1,2$. This
%comes form the observation that to achieve
%\eqref{eq_BE_BC_UB_point_2}, in the third phase, the number of
%layers used by $\msf{B}_i$ is $n_{i\msf{R}}$. To have different
%MSEs at $\msf{B}_1$ and $\msf{B}_2$, we are motivated to use
%successive refinement \cite{cheng2005successive} for the source
%coding part.

Our scheme is summarized as follows \\
\noindent \textbf{Phase I}: By using random Gaussian codebook,
relay sends coded symbols $X_{\msf{R}}[t]$, $t=1 \ldots T$
    from the codeword representing message for user $\msf{B}_1$.\par
\noindent \textbf{Phase II}: Relay sends
$X_{\msf{R}}[t]=X_{2\msf{R}}[t]+X_{1\msf{R}}[t]$, $t=T+1 \ldots
    2T$, where $X_{2\msf{R}}[t]$ are coded symbols for user $\msf{B}_2$ and
\begin{equation} \label{eq_superposition_x1}
X_{1\msf{R}}[t]=\left(C_{1\msf{R}}[t]-w \cdot
h_{\msf{B}_1\msf{R}}X_{2\msf{R}}[t] - d[t]  \right) \mod \; L
\end{equation}
where similar to \eqref{eq_ul_Lattice}, $C_{1\msf{R}}[t]$ is coded
symbol for user $\msf{B}_1$. $d[t]$ is the independent dither. For
a real number $x, x \mod \;L=x-Q_L(x)$ with $Q_L(x)$ being the
nearest multiple of $L$ to $x$. \par \noindent \textbf{Phase III}:
Let the erased symbols sent to the wrong receiver in Phase 1 and 2
be $X^\mathrm{I}_{\msf{e}}$ and $X^\mathrm{II}_{\msf{e}}$
respectively. Relay first quantizes the length $Tp(1-p)$ sequence
$X^\mathrm{I}_{\msf{e}}+X^\mathrm{II}_{\msf{e}}$ using successive
refinement into indexes $i_\msf{c}$ and $i_\msf{r}$, where
$i_\msf{c}$ is the common index which will be decoded for both
$\msf{B}_1$ and $\msf{B}_2$ while $i_\msf{r}$ is the refinement
index which will be decoded only at $\msf{B}_1$. Gaussian
superposition channel coding with length $T(1-p)$ is adopted to
transmit $(i_\msf{c}$,$i_\msf{r})$.
%Note that we have on/off
%channel statistics with (on) probability $p \leq 1$. It is then
%crucial to choose the length of the channel code smaller than that
%of the source code.

Now, user $\msf{B}_2$ can know the noisy reconstruction
$X^\mathrm{I}_{\msf{e}}+X^{\mathrm{II}}_{\msf{e}}+Z_{\msf{D}2}$
with MSE $\msf{D}_2$, by decoding $i_\msf{c}$. With proper power
allocation, $\msf{B}_1$ (better channel) knows the reconstruction
$X^\mathrm{I}_{\msf{e}}+X^{\mathrm{II}}_{\msf{e}}+Z_{\msf{D}1}$ by
successively decoding $i_\msf{c}$ and $i_\msf{r}$, where
reconstruction error $Z_{\msf{D}1}$ has smaller MSE $\msf{D}_1$
than that of $Z_{\msf{D}2}$. For this two-receiver source-channel
coding, the rates for common index $i_\msf{c}$ and refinement
index $i_\msf{r}$ are chosen as
\begin{equation}
\textstyle\log\left(1+\frac{2}{\msf{D}_{2}}\right)\label{eq_Main_common_decode}\
\text{and}\
%\end{align}
%and
%\begin{equation} \label{eq_Main_refine_decode}
\textstyle\log\left(1+\frac{2}{\msf{D}_{1}}\right) - \log\lp
1+\frac{2}{\msf{D}_{2}}\rp
\end{equation}
respectively. To ensure successful channel decoding at receivers,
we need to carefully choosing the power allocation of the
superposition channel coding, as well as $\msf{D}_1$ and
$\msf{D}_2$ in \eqref{eq_Main_common_decode}. Let the power
allocation for indexes $i_\msf{c}$ and $i_\msf{r}$ be
$\msf{SNR_c}$ and $\msf{SNR_r}$ respectively. We choose
$\msf{SNR_r} =1/\msf{SNR_{2R}},\msf{SNR_{c}=1-SNR_r}$. (Here we
only provide the proof when $\msf{SNR}_{2\msf{R}}\geq 2$, since
the bounded-gap result for $\msf{SNR}_{2\msf{R}} < 2$ is trivial.)
For $\msf{B}_2$ to correctly decode $i_\msf{c}$, from
\eqref{eq_Main_common_decode} and the lengths of channel and
source codes, we need to choose
\begin{equation} \label{eq_Main_SR_D2}
\textstyle \msf{D}_{2}=\frac{4}{\msf{SNR_{2R}}-1}
\end{equation}
For receiver $\msf{B}_1$ to decode both $i_\msf{c}$ and
$i_\msf{r}$, we choose
\begin{equation} \label{eq_Main_SR_D1}
\textstyle \msf{D}_1=\frac{4}{\msf{SNR_{1R}}+\msf{SNR_{2R}}}.
\end{equation}
Note that $\msf{D}_i$ is inverse proportional to
$\msf{SNR}_{i\msf{R}}, i=1,2$, consistent with the insights from
the binary expansion model.

Now receivers $\msf{B}_i, i=1,2$ can obtain reconstructions of
erased symbols $X^\mathrm{I}_{\msf{e}}+X^{\mathrm{II}}_{\msf{e}}$
with $\msf{D_1}$ in \eqref{eq_Main_SR_D1} and $\msf{D_2}$ in
\eqref{eq_Main_SR_D2} respectively. With side-information
$X^\mathrm{I}_{\msf{e}}$ (noisy) from Phase I, $\msf{B}_2$ can
combine $Tp(1-p)$ reconstructed symbols
$X^{\mathrm{II}}_{\msf{e}}+Z_{\msf{D}2}-Z_{\msf{B}_2}$ and the
$Tp$ un-erased symbols received in Phase II to decode XOR
$\Sigma_2$. Then \eqref{eq_AWGN_BC_UB_point_2_R2} is achievable
with bounded gap $\Delta_2$. To see this, we can first upper-bound
the MSE distortion $\msf{D}_2$ in \eqref{eq_Main_SR_D2} as
\begin{equation} \label{leq_Main_SR_D2_ub}
\textstyle \msf{D}_{2} \leq \frac{8}{\msf{SNR_{2R}}}.
\end{equation}
By choosing the power allocation of $X_\msf{1R}$ and $X_\msf{2R}$
in Phase II be $1/\msf{SNR_{2R}}$ and $1-1/\msf{SNR_{2R}}$
respectively, together with the independence of these two signals,
we have the following achievable rate for user $\msf{B}_2$
\begin{align}
R_{2} \geq \frac{1}{(3-p)}\Bigg(&
p(1-p)\C\bigg(\frac{1-\frac{1}{\msf{SNR_{2R}}}}{\frac{8}{\msf{SNR_{2R}}}+\frac{1}{\msf{SNR_{2R}}}+\frac{1}{\msf{SNR_{2R}}}}\bigg)\notag
\\
&+p\;\C\bigg(\frac{1-\frac{1}{\msf{SNR_{2R}}}}{\frac{1}{\msf{SNR_{2R}}}+\frac{1}{\msf{SNR_{2R}}}}\bigg)\Bigg)
\label{eq_Main_R2_general_1}
\end{align}
where \eqref{leq_Main_SR_D2_ub} is applied to obtain the first
term in the RHS of \eqref{eq_Main_R2_general_1}. Then bounded gap
result can be obtained from \eqref{eq_Main_R2_general_1}.
%Note that equivalently, the erasure probability for $\msf{B}_2$ is
%decreased from $1-p$ to $1-p \cdot \frac{2-p}{3-p}$ from delayed
%state information.

Now we show that for user $\msf{B}_1$, rate $R_1-\Delta_1$, with
$R_1$ in \eqref{eq_AWGN_BC_UB_point_2_R1}, is achievable.
Following similar procedure as $\msf{B}_2$ aforementioned, by
combining the erased symbols with the un-erased symbols received
in Phase I, the following rate is achievable to decode XOR
$\Sigma_1$,
\begin{align}\textstyle
     \frac{p(2-p)}{3-p}\C(\msf{SNR_{1R}})-\frac{p(1-p)}{3-p}\log(3). \label{eq_Main_R1_general}
\end{align}
where the following inequality from \eqref{eq_Main_SR_D1} is used
\[\textstyle
\msf{D}_1 \leq \frac{2}{\msf{SNR_{1R}}}.
\]
Moreover, user $\msf{B}1$ can decode additional messages by
forming the following channel from the un-erased symbols in Phase
II,
\begin{align}
\left[C_{1\msf{R}}[t] + E_L(t) \right] \mod \; L,
\label{eq_Main_mod_A}
\end{align}
where $E_L(t)=(wh_{\msf{B_1R}}-1)X_{1\msf{R}}[t]+wZ_\msf{B_1}[t]$.
The channel \eqref{eq_Main_mod_A} is a modulo-$L$ channel with
length $Tp$ and power $L^2/12=1/\msf{SNR_{2R}}$, then rate
\begin{align} \label{eq_Main_R1_phase2_rate}
\textstyle
\frac{p}{3-p}\left(\log\lp\frac{\msf{SNR_{1R}}}{\msf{SNR_{2R}}}\rp-\log(2
\pi e /12)\right)
\end{align}
is achievable. By summing \eqref{eq_Main_R1_phase2_rate} and
\eqref{eq_Main_R1_general}, our achievable rate for user
$\msf{B}_1$ has bounded gap to \eqref{eq_AWGN_BC_UB_point_2_R1}

%\renewcommand{\baselinestretch}{0.85}

%\newpage
%\renewcommand{\baselinestretch}{1}
\appendix

\subsection{Detailed proof of the Inner Bound $\mscr{R}^{\uplink}_{\G,\Inn}\lp\mathrm{d}\rp$  in
\eqref{eq_Gau_CapIn}}

Here we focus on the bounded-gap achievability to the sum rate in
\eqref{eq_up_delayedub_C3} by joint lattice decoding. From
\eqref{eq_ul_Lattice_model} and \eqref{eq_ul_Lattice}, it can be
easily shown that the post-processed signal
\[
\left( \mathbf{Wy}_{\msf{R}}+
[(\mathbf{d}_{\msf{A}_1}+\mathbf{d}_{\msf{B}_1})^T \;\;
(\mathbf{d}_{\msf{A}_2}+\mathbf{d}_{\msf{B}_2})^T]^T \right) \;
\mathrm{mod} \; (\Lambda_S \times \Lambda_S).
\]
equals to
\begin{align}
\left( \left[(\mathbf{c}_{A_1}+\mathbf{c}_{B_1} )^T\;\;
(\mathbf{c}_{A_2}+\mathbf{c}_{B_2} )^T\right]^T+ \mathbf{E}\right)
\; \mathrm{mod}  \; (\Lambda_S \times \Lambda_S),
\label{eq_uL_eq_channel_1}
\end{align}
where
\begin{equation} \label{eq_uL_eq_E}
\mathbf{E}=\left(\mathbf{W}\mathbf{H}-\mathbf{I}\right)\left[(\mathbf{x}_{\msf{A}_1}+\mathbf{x}_{\msf{B}_1})^T\;
\;
(\mathbf{x}_{\msf{A}_2}+\mathbf{x}_{\msf{B}_2})^T\right]^T+\mathbf{Wz}_\msf{R},
\end{equation}
and $\mathbf{H}=\left[\mathbf{H}_{\msf{A}_1} \;
\mathbf{H}_{\msf{A}_2}\right]$. Moreover, in
\eqref{eq_uL_eq_channel_1}, the sum of user codewords within a
pair $i=1,2$ is still a lattice codeword
\[(\mathbf{c}_{\msf{A}i}+\mathbf{c}_{\msf{B}i}) \; \mathrm{mod} \;
\Lambda_S \in \Lambda_{i}.
\]
Then from \cite{scTWCOM14}, the following sum rate is achievable
by jointly lattice decoding
\begin{equation} \label{eq_uL_eq_sum_rate0}
R_1+R_2 \geq \mathop{\lim}\limits_{T \rightarrow
\infty}\frac{1}{2T} \log
\left(\frac{|\frac{1}{2}\mathbf{I}|}{|\Sigma_E|}\right),
\end{equation}
where $\Sigma_E$ is the covariance matrix of $\mathbf{E}$ in
\eqref{eq_uL_eq_E}. With $\mathbf{W}$ chosen as the MMSE filter,
the information lossless property of MMSE estimation can be
invoked, then the sum rate in \eqref{eq_uL_eq_sum_rate0} becomes
\begin{equation} \label{eq_uL_eq_sum_rate1}
R_1+R_2 \geq \mathop{\lim}\limits_{T \rightarrow
\infty}\frac{1}{2T} \log
\left(\left|\frac{1}{2}\mathbf{I}+\mathbf{H}\mathbf{H}^T\right|\right).
\end{equation}
Now from \eqref{eq_Gau_up_HA}
($\mathbf{H}=\left[\mathbf{H}_{\msf{A}_1} \;
\mathbf{H}_{\msf{A}_2}\right]$) and the ergodicity of state
process,
\begin{align}
R_1+R_2 \geq &\mathrm{E}_{S_1,S_2}\left[\C \lp S_1
\msf{SNR_{R1}}+S_2
\msf{SNR_{R2}}\rp  \right]-1 \notag \\
\geq & \; p(1-p) \lp \C \lp \msf{SNR_{R1}}+ \C \msf{SNR_{R2}}\rp
\rp \notag
\\ &+p^2 \C\lp \msf{SNR_{R1}}+ \msf{SNR_{R2}}\rp -1 \label{eq_uL_eq_sum_rate2}
\end{align}
To compared \eqref{eq_uL_eq_sum_rate2} with
\eqref{eq_up_delayedub_C3}, it can be easily checked that
\begin{align}
&\C\lp \msf{SNR_{R1}}+
\msf{SNR_{R2}+2\sqrt{\msf{SNR_{R1}}\msf{SNR_{R2}}}}\rp\notag
\\ \leq & 1+\C\lp \msf{SNR_{R1}}+ \msf{SNR_{R2}}\rp.
\end{align}
Then the bounded-gap result for sum rate
\eqref{eq_up_delayedub_C3} is established. The bounded-gap
achievement to the RHSs of \eqref{eq_up_delayedub_C1} follows
similarly. As a final note, decoding correct
$(\mathbf{c}_{\msf{A}i}+\mathbf{c}_{\msf{B}i}) \; \mathrm{mod} \;
\Lambda_S$ from aforementioned joint lattice decoding equals to
decode corrects XORs $\Sigma_i=W_{\msf{A}_i} \oplus W_{\msf{B}_i},
i=1,2$ \cite{Nazer}.

\subsection{Detailed proof of the inner bound $\mscr{R}^{\downlink}_{\G,\Inn}\lp\mathrm{d}\rp$  in
\eqref{eq_Gau_CapIn}}

Here we provide the detailed proof for achieving the Gaussian
downlink rate region with delayed state information
$\mscr{R}^{\downlink}_{\G,\Inn}\lp\mathrm{d}\rp$. We assume that
state $(S_1^{t-1},S_2^{t-1})$ are known at both receivers
$\msf{B}_1$ and $\msf{B}_2$ at time $t$. In the following, we only
provide the proof when $\msf{SNR}_{2\msf{R}}\geq 2$, since the
bounded-gap result for $\msf{SNR}_{2\msf{R}} < 2$ is trivial. We
first focus on the proposed three-phase scheme to achieve
$(R_1-\Delta_1,R_2-\Delta_2)$ from
\eqref{eq_AWGN_BC_UB_point_2_R1}\eqref{eq_AWGN_BC_UB_point_2_R2},
where $(\Delta_1,\Delta_2)$ are given in
\eqref{eq_Distortion1}\eqref{eq_Distortion2}. Here we give the
detailed definitions of the erased symbol sequences
$X^\msf{I}_\msf{e}$ and $X^\msf{II}_\msf{e}$ in Phase III. First,
the $X_\msf{R}[t]$, $t=1 \ldots T$ in Phase I forms a codeword
from a random Gaussian codebook to encode $\Sigma_1$ for
$\msf{B}_1$. In Phase III, from the delayed state information, the
relay knows the erased indexes $t$s where state sequence with
length $T_1$ of which $(S_{1}[t],S_2[t])=(0,1)$ in Phase I, $
1\leq t \leq T$. For each $X_\msf{R}[t]$, we assign it to an
unique symbol $X^\msf{I}_\msf{e}[t']=X_\msf{R}[t]$ in the erased
symbol sequence. If $T_1>Tp(1-p)$, then we abandon the last
$T_1-Tp(1-p)$ symbols. On the contrary, if $T_1<Tp(1-p)$, we set
$X_{\msf{e}}^\msf{I}[t']=0, t'=T_1+1, \ldots Tp(1-p)$. Then the
total length of erased symbol sequence $X_{\msf{e}}^\msf{I}[t']$
in Phase I is $Tp(1-p)$. $B_1$ also knows the mapping from $t$ to
$t'$ as $\Pi_\msf{I}$, where $t=\Pi_\msf{I}(t')$ with $t'=1 \ldots
\min\{T_1,Tp(1-p)\}, 1 \leq t \leq T$. The mapping $\Pi_\msf{I}$
is also known at receiver $\msf{B}_1$. In Phase II,
$X_{2\msf{R}}[t]$, $t=T+1, \ldots ,2T$ (independent of
$X_{1\msf{R}}[t]$ in \eqref{eq_superposition_x1}) forms a codeword
from a Gaussian codebook to encode $\Sigma_2$ for $\msf{B}_2$. And
as aforementioned, we can form length $Tp(1-p)$ erased symbol
sequence in Phase II
$X_{\msf{e}}^\msf{II}[\Pi_\msf{II}(t')]=X_{1\msf{R}}[t]+X_{2\msf{R}}[t]$
from the state sequence with length $T_2$ of which
$(S_{1}[t],S_{2}[t])=(1,0)$, $T+1 \leq t \leq 2T$. The
corresponding mapping $\Pi_\msf{II}$ is known at receiver
$\msf{B}_2$.

In Phase III, the relay compresses the sum of erased symbol
sequence $X_{\msf{e}}^\msf{I}[t']+X_{\msf{e}}^\msf{II}[t'], t'=1
\ldots Tp(1-p)$ with successive refinement, and send the
corresponding quantization indexes $i_\msf{c}$ and $i_\msf{r}$
through the on/off downlink. At receiver $\msf{B}_i, i=1,2$, with
successfully decoding the quantization index(es), we wish to
obtain the following reconstruction
\begin{equation} \label{eq_SuccRecon}
X_{\msf{e}}^\msf{I}[t']+X_{\msf{e}}^\msf{II}[t']+Z_{\msf{D}i}[t']
\end{equation}
where $Z_{\msf{D}i} \sim \mcal{CN}\lp 0,\msf{D}_i \rp$ and
$\msf{D}_1 \leq \msf{D}_2$. To solve this two-receiver
source-channel coding problem, we need to carefully select the
lengths and rates of the source and channel codes. The length of
source code is $Tp(1-p)$, and the rate selection of the source
code comes as follows. To validate \eqref{eq_SuccRecon}, we form
the following test channels for successive refinement
\begin{align}
&V_1=X_\msf{sum}+Z; \label{eq_SR_test_B1} \\
&V_2=V_1+Z'=X_\msf{sum}+Z+Z', \label{eq_SR_test_B2}
\end{align}
with source $X_\msf{sum}$ having variance 2 and same distribution
as $X_{\msf{e}}^\msf{I}[t']+X_{\msf{e}}^\msf{II}[t']$;
reconstruction errors $Z \sim \mcal{CN} \lp 0,\msf{D}_1 \rp$ and
$Z' \sim \mcal{CN} \lp 0, \msf{D}_2-\msf{D}_1 \rp$. Here
$X_\msf{sum}, Z$ and $Z'$ are independent. Note that due to the
$\mod L$ operation in \eqref{eq_superposition_x1}, the source
$X_\msf{sum}$ is not Gaussian. We choose the rate for common index
$i_\msf{c}$ from \eqref{eq_SR_test_B2} as
\begin{align}
&\log\left(\frac{2+(\msf{D}_2-\msf{D}_1)+\msf{D}_1}{(\msf{D}_2-\msf{D}_1)+\msf{D}_1}\right)=\log\left(1+\frac{2}{\msf{D}_{2}}\right)\notag
\\ \geq & I(X_\msf{sum};V_{2}), \label{eq_common_decode}
\end{align}
where the inequality comes from that Gaussian source is the
hardest one to quantize \cite{BooK_NIT_KIM}. The rate for
refinement index $i_\msf{r}$ is chosen as
\begin{equation} \label{eq_refine_decode}
\log\!\left(\frac{1+\frac{2}{\msf{D}_{1}}}{1+\frac{2}{\msf{D}_{2}}}\right)
\geq I(X_\msf{sum};V_{1})\!-\!I(X_\msf{sum};V_{2}).
\end{equation}
Note that the RHS is similar to the achievable rate in a Gaussian
wiretap channel. Thus from \cite{leung1978gaussian}, we know that
for any source with variance $2$, the Gaussian source $\mcal{CN}
(0,2)$ maximizes the RHS. Note that $V_{2} \rightarrow V_{1}
\rightarrow X_\msf{sum}$. By choosing $V_1$ as the reconstruction
distribution at $\msf{B}_1$ while $V_2$ as that at $\msf{B}_2$,
from \eqref{eq_SR_test_B1}-\eqref{eq_refine_decode}, the MSE
distortion at $\msf{B}_1$ is $\msf{D}_1$ while that at $\msf{B}_2$
is $\msf{D}_2$ respectively \cite{cheng2005successive}.

 Now we must ensure that the receiver $\msf{B}_1$ can
correctly decode both $i_\msf{c}$ and $i_\msf{r}$ while receiver
$\msf{B}_2$ can correctly decode $i_\msf{c}$. This task is done by
carefully choosing the length and power allocation of the
superposition channel coding, as well as $\msf{D}_1$ and
$\msf{D}_2$ in \eqref{eq_common_decode}\eqref{eq_refine_decode}.
First, since the channel has on/off probability $p$, it is crucial
to choose the length of channel code longer than that of source
code, which is $T(1-p)$ for the length $Tp(1-p)$ source code. Now
indexes $i_\msf{c}$ and $i_\msf{r}$ are channel encoded using
independent Gaussian codebooks with power $\msf{SNR_c}$ and
$\msf{SNR_r}$ respectively, with power allocation $\msf{SNR_r}
=1/\msf{SNR_{2R}},\msf{SNR_{c}=1-SNR_r}$. For $\msf{B}_2$ to
correctly decode $i_\msf{c}$, from \eqref{eq_common_decode} and
the lengths of channel and source codes, we need
\begin{align*}
&Tp(1-p)\log\left(1+\frac{2}{\msf{D}_{2}}\right)
\\ \leq&T(1-p)\cdot
p\;\C\left(\frac{\msf{SNR_{2R}}\left(1-\frac{1}{\msf{SNR_{2R}}}\right)}{\msf{SNR_{2R}}\frac{1}{\msf{SNR_{2R}}}+1}\right),
\end{align*}
which results in
\begin{equation} \label{eq_superposition_R2}
\log\left(1+\frac{2}{\msf{D}_{2}}\right)\leq
\log\left(\frac{1+\msf{SNR_{2R}}}{2}\right).
\end{equation}
Then we can choose
\begin{equation} \label{eq_SR_D2}
\msf{D}_{2}=\frac{4}{\msf{SNR_{2R}}-1}
\end{equation}
For receiver $\msf{B}_1$, first note that since $\msf{SNR_{1R}}
\geq \msf{SNR_{2R}}$, then $\msf{B}_1$ can also successfully
decode common index $i_\msf{c}$ by treating the codeword for
$i_\msf{r}$ as noise. After subtracting the codewords
corresponding to $i_\msf{c}$, from \eqref{eq_refine_decode} and
the lengths of channel and source codes, we need
\begin{equation} \label{eq_superposition_R1}
\log\lp 1+\frac{2}{\msf{D}_{1}} \rp-\log \lp
1+\frac{2}{\msf{D}_{2}} \rp
\leq\C\left(\frac{\msf{SNR_{1R}}}{\msf{SNR_{2R}}}\right)
\end{equation}
to correctly decode refinement index $i_\msf{r}$. From
\eqref{eq_SR_D2}, we must choose $\msf{D}_1$ satisfying
\[
\log \lp 1+\frac{2}{\msf{D}_{1}} \rp\leq
\log\left(\frac{1+\msf{SNR_{2R}}}{2}\right)+\log\left(1+\frac{\msf{SNR_{1R}}}{\msf{SNR_{2R}}}\right),
\]
which is equivalent to
\[
1+\frac{2}{\msf{D}_{1}} \leq \frac{1}{2}
\left(1+\frac{\msf{SNR_{1R}}}{\msf{SNR_{2R}}}+\msf{SNR_{2R}}+\msf{SNR_{1R}}\right).
\]
Because $\msf{SNR_{1R}}/\msf{SNR_{2R}} \geq 1$, the above
inequality can be meet if
\begin{equation} \label{eq_SR_D1}
\msf{D}_1=\frac{4}{\msf{SNR_{1R}}+\msf{SNR_{2R}}}
\end{equation}

Now receivers $\msf{B}_i, i=1,2$ can obtain reconstructions
\eqref{eq_SuccRecon} with $\msf{D_1}$ in \eqref{eq_SR_D1} and
$\msf{D_2}$ in \eqref{eq_SR_D2} respectively, where $\msf{D_1}$ in
\eqref{eq_SR_D2} is smaller than $\msf{D_2}$ in \eqref{eq_SR_D2}.
Note that the (noisy) erased sequence in Phase I
$X_{\msf{e}}^\msf{I}[t']$ is known at $\msf{B}_2$. As a
side-information, $\msf{B}_2$ can subtract
$X_{\msf{e}}^\msf{I}[t']$ from reconstructions
\eqref{eq_SuccRecon} and obtain
$X_{\msf{e}}^\msf{II}[t']+Z_{\msf{D}2}[t']$ (with additional
channel noise). Now $\msf{B}_2$ can combine the erased
$X_{\msf{e}}^\msf{II}[t']+Z_{\msf{D}2}[t']$ with the un-erased
symbols received in Phase II to decode XOR $\Sigma_2$, with
bounded gap to \eqref{eq_AWGN_BC_UB_point_2_R2} as
\begin{align}
R_{2}\geq
\frac{p(2-p)}{3-p}\C(\msf{SNR_{2R}})-\frac{p(1-p)}{3-p}\log(10)-\frac{p}{3-p}.
\label{eq_R2_general}
\end{align}
The details come as follows. First, we can upper-bound the MSE
distortion $\msf{D}_2$ in \eqref{eq_SR_D2} as
\begin{equation} \label{leq_SR_D2_ub}
\msf{D}_{2}=\frac{4}{\msf{SNR_{2R}}-1} \leq
\frac{4}{\frac{1}{2}\msf{SNR_{2R}}}=\frac{8}{\msf{SNR_{2R}}}.
\end{equation}
The above inequality is due to that the SNR regime we considered $
\msf{SNR_{2R}} \geq 2$ is equivalent to $\msf{SNR_{2R}}-1 \geq
\frac{1}{2}\msf{SNR_{2R}}$. And from test channel
\eqref{eq_SR_test_B2}, it is ensured that the quantization noise
$Z_{\msf{D}2}[t']$ is independent of $X_{\msf{e}}^\msf{II}[t']$.
Now we can from the following sequence to decode $\Sigma_2$ at
receiver $\msf{B}_2$
\[
\tilde{Y}_{\msf{B}_2}[t]\!=\!\left\{\!
\begin{array}{ll}
X_\msf{R}[\Pi_\msf{II}^{-1}(t')]+Z_{\msf{D}_2}[\Pi^{-1}_\msf{II}(t')]-\frac{Z_{\msf{B_2}}[t]}{h_\msf{B_2R}} & t= \Pi_\msf{II}(t')  \\
S_2[t]h_\msf{B_2R}X_\msf{R}[t]+Z_\msf{B_2}[t] & \mbox{other} \;t,
\end{array} \right.
\]
where $X_\msf{R}[t]=X_\msf{2R}[t]+X_\msf{1R}[t]$ with
$X_\msf{2R}[t]$ carrying message $\Sigma_2$ and interference
$X_\msf{1R}[t]$ from \eqref{eq_superposition_x1}, $T+1 \leq t \leq
2T$, $t'=1 \ldots \min\{T_1,T_2,Tp(1-p)\}$. Note that when $T
\rightarrow \infty$, $T_1/T,T_2/T \rightarrow p(1-p)$ almost
surely from our random state process. The power allocations of
$X_\msf{1R}$ and $X_\msf{2R}$ are $1/\msf{SNR_{2R}}$ and
$1-1/\msf{SNR_{2R}}$. When $T \rightarrow \infty$, from the
independence of $X_\msf{2R}[t]$ and $X_\msf{1R}[t]$ and our power
allocation, we have the following achievable rate for user
$\msf{B}_2$
\begin{align}
R_{2} \geq
\frac{1}{(3-p)}\Bigg(&p(1-p)\C\bigg(\frac{1-\frac{1}{\msf{SNR_{2R}}}}{\frac{8}{\msf{SNR_{2R}}}+\frac{1}{\msf{SNR_{2R}}}+\frac{1}{\msf{SNR_{2R}}}}\bigg)\notag
\\
&+p\;\C\bigg(\frac{1-\frac{1}{\msf{SNR_{2R}}}}{\frac{1}{\msf{SNR_{2R}}}+\frac{1}{\msf{SNR_{2R}}}}\bigg)\Bigg)
\label{eq_R2_general_1}
\end{align}
where \eqref{eq_R2_general_1} comes from \eqref{leq_SR_D2_ub} and
the fact that Gaussian interference is the worst interference
under the same power constraint $1/\msf{SNR_{2R}}$. Then
\eqref{eq_R2_general} can be easily obtained from
\eqref{eq_R2_general_1}.

Now we show that for user $\msf{B}_1$, bounded-gap rate
$R_1-\Delta_1$ with $R_1$ in \eqref{eq_AWGN_BC_UB_point_2_R1} and
$\Delta_1$ in \eqref{eq_Distortion1}, is achievable. Following
similar procedure as $\msf{B}_2$ aforementioned, by combining the
erased symbols with the un-erased symbols received in Phase I, the
following rate is achievable to decode XOR $\Sigma_1$,
\begin{align}
&\frac{1}{(3-p)}\left(p(1-p)\C\left(\frac{1}{\frac{2}{\msf{SNR_{1R}}}+\frac{1}{\msf{SNR_{1R}}}}\right)+p
\; \C(\msf{SNR_{1R}})\right) \label{eq_R1_general_1} \\ \geq &
     \frac{p(2-p)}{3-p}\C(\msf{SNR_{1R}})-\frac{p(1-p)}{3-p}\log(3). \label{eq_R1_general}
\end{align}
where the following inequality (from \eqref{eq_SR_D1}) is used for
obtaining \eqref{eq_R1_general_1}
\[
\msf{D}_1=\frac{4}{\msf{SNR_{1R}}+\msf{SNR_{2R}}} \leq
\frac{2}{\msf{SNR_{1R}}},
\]
which follows from the assumption $\msf{SNR_{1R}} \geq
\msf{SNR_{2R}}$. Moreover, user $\msf{B}1$ can decode additional
messages by forming the following channel from the un-erased
symbols in Phase II,
\begin{align}
&\left\{ w \left[
h_{\msf{B_1R}}\left(X_{1\msf{R}}[t]+X_{2\msf{R}}[t]\right)+Z_\msf{B_1}[t]\right]+
d(t) \right\} \mod \; L
\notag \\
=&\left[C_{1\msf{R}}[t] + E_L(t) \right] \mod \; L,
\label{eq_mod_A}
\end{align}
where $E_L(t)=(wh_{\msf{B_1R}}-1)X_{1\msf{R}}[t]+wZ_\msf{B_1}[t]$.
The channel \eqref{eq_mod_A} is a modulo-$L$ channel with length
$Tp$ and power $L^2/12=1/\msf{SNR_{2R}}$. By choosing $w$ as the
MMSE coefficient, rate
\begin{align} \label{eq_R1_phase2_rate}
\frac{p}{3-p}\left(\log\lp\frac{\msf{SNR_{1R}}}{\msf{SNR_{2R}}}\rp-\log(2
\pi e /12)\right)
\end{align}
is achievable \cite{nested_lattice}. To compare with
\eqref{eq_AWGN_BC_UB_point_2_R1}, note that the RHS of
\eqref{eq_AWGN_BC_UB_point_2_R1} can be rewritten as
\begin{align}
&\frac{1}{3-p}\left(p\left(\C \lp \msf{SNR_{1R}} \rp -\C \lp
\msf{SNR_{2R}} \rp\right)+p(2-p)\C\lp\msf{SNR_{1R}} \rp \right) \notag \\
 \leq &
\frac{1}{3-p}\left(p\log\left(\frac{\msf{SNR_{1R}}}{\msf{SNR_{2R}}}\right)+
p(2-p)\C\lp\msf{SNR_{1R}} \rp \right),
\end{align}
where the second inequality is due to assumption $\msf{SNR_{1R}}
\geq \msf{SNR_{2R}} $. By summing \eqref{eq_R1_phase2_rate} and
\eqref{eq_R1_general}, our achievable rate for user $\msf{B}_1$
has bounded gap to \eqref{eq_AWGN_BC_UB_point_2_R1} as
\begin{align}
R_1 \geq
&p\left(\C(\msf{SNR}_{1\msf{R}})\!-\!\C(\msf{SNR}_{2\msf{R}})\right)\!+\!
\frac{p(2-p)}{3-p}\C(\msf{SNR}_{2\msf{R}}) \notag
\\&-\frac{p(1-p)}{3-p}\log(3)-\frac{p}{3-p}\log \lp \frac{2\pi e}{12} \rp.
\notag
\end{align}

Finally, the intersection of \eqref{eq_AWGN_BC_UB_R2} and
\eqref{eq_AWGN_BC_UB_R2_1} is
\begin{equation} \label{eq_AWGN_BC_UB_point_1}
(R_1,R_2)=\Big(p\left(\C(\msf{SNR}_{1\msf{R}})\!-\!\C(\msf{SNR}_{2\msf{R}}\right),p\C(\msf{SNR}_{2\msf{R}}\Big).
\end{equation}
For this point, we use Gaussian superposition coding to transmit
messages $\Sigma_1$ and $\Sigma_2$ with corresponding power
allocations $1/\msf{SNR_{2R}}$ and $1-1/\msf{SNR_{2R}}$. Then we
have the following achievable rate pair $(R_1,R_2)$
\[
\Big(p\left(\C(\msf{SNR}_{1\msf{R}})\!-\!\C(\msf{SNR}_{2\msf{R}}\right),p\C(\msf{SNR}_{2\msf{R}})-p
\Big),
\]
which has bounded gap to \eqref{eq_AWGN_BC_UB_point_1}. And it
concludes our proof.

\subsection{Detailed proof of the outer-bound region in
\eqref{eq_Gau_CapOuter}}

Here we prove that with delayed state information, the region
$\mscr{R}^{\downlink}_{\Out}\lp\mathrm{d}\rp$ defined from
\eqref{eq_AWGN_BC_UB_R2_1}\eqref{eq_AWGN_BC_UB_R1}\eqref{eq_AWGN_BC_UB_R2}
is a outer-bound region for the Gaussian model
\eqref{eq_Gau_DL_channel}\eqref{eq_Gau_UP_channel}. We first focus
on \eqref{eq_AWGN_BC_UB_R2}, which results from first forming an
equivalent degraded downlink and then outer-bounding carefully to
avoid explicitly selecting the auxiliary random variable $U$.
Since the noises at receivers in \eqref{eq_Gau_DL_channel} are
independent of the delayed state feedback, we can change
$Y_{\msf{B}_1}$ and $Y_{\msf{B}_2}$ as
\begin{align}
&Y_{\msf{B}_1}=S_{1}X_{\msf{R}}+\frac{Z_{\msf{B}_1}}{h_{\msf{B_1R}}} \label{eq_Y_B1_neq_gain}\\
&Y_{\msf{B}_2}=S_{2}X_{\msf{R}}+\frac{Z_{\msf{B}_2}}{h_{\msf{B_2R}}}=S_{2}X_{R}+\frac{Z_{\msf{B}_1}}{h_{\msf{B_1R}}}+Z'
\label{eq_Y_B2_neq_gain}
\end{align}
where $Z' \sim
\mcal{CN}(0,\frac{1}{\msf{SNR}_{2\msf{R}}}-\frac{1}{\msf{SNR}_{1\msf{R}}})$.
By giving $Y_{\msf{B}_2}$ in \eqref{eq_Y_B2_neq_gain} to
$\msf{B}_1$, we have physically degraded channel
\cite{BooK_NIT_KIM} $X_\msf{R} \rightarrow
\{Y_{\msf{B}_1},Y_{\msf{B}_2}\}\rightarrow Y_{\msf{B}_2}$ with the
following outer bounds
\begin{align}
R_{2}&\leq I(U;Y_{\msf{B}_2}|S) \label{eq_Gau_R2_UB_1} \\
&=I(X_\msf{R},U;Y_{\msf{B}_2}|S_2)-I(X_\msf{R};Y_{\msf{B}_2}|U,S_2) \notag \\ &=I(X_\msf{R};Y_{\msf{B}_2}|S_2)-pr \label{eq_Gau_R2_UB_2} \\
&=p \cdot \C(\msf{SNR_{2R}})-pr,\label{eq_Gau_R2_UB_3}
\end{align}
where $S=\{S_1,S_2\}$ in \eqref{eq_Gau_R2_UB_1}, and
\eqref{eq_Gau_R2_UB_2} comes from conditional Markov Chain $U
\rightarrow X_\msf{R} \rightarrow Y_{\msf{B}_2}$ given $S_2$, and
$r$ is defined as
\begin{equation} \label{eq_Gau_R2_UB_r}
r \triangleq I(X_\msf{R};Y_{\msf{B}_2}|U,S_2=1);
\end{equation}
also
\begin{align}
R_{1}\leq & I(X_\msf{R};Y_{\msf{B}_1},Y_{\msf{B}_2}|S,U) \notag \\
=&I(X_\msf{R};Y_{\msf{B}_1}|S,U)+I(X_\msf{R};Y_{\msf{B}_2}|Y_{\msf{B}_1},S,U) \notag\\
=&pI(X_\msf{R};Y_{\msf{B}_1}|U,S_{1}=1)+p(1-p)I(X_\msf{R};Y_{\msf{B}_2}|U,S_{2}=1)\notag
\\&+p^{2}I(X_\msf{R};Y_{\msf{B}_2}|Y_{\msf{B}_1},U,S_{1}=S_{2}=1). \label{eq_Gau_R1_UB_1}
\end{align}
From \eqref{eq_Y_B1_neq_gain} and \eqref{eq_Y_B2_neq_gain},
\[
I(X_\msf{R};Y_{\msf{B}_2}|Y_{\msf{B}_1},U,S_{1}=S_{2}=1)=0,
\]
then
\begin{align}
R_{1}\leq pI(X_\msf{R};Y_{\msf{B}_1}|U,S_{1}=1)+p(1-p)r
\label{eq_Gau_R1_UB_0}
\end{align}
With $r$ defined in \eqref{eq_Gau_R2_UB_r},
\begin{align}
&I(X_\msf{R};Y_{\msf{B}_1}|U,S_{1}=1)-r \notag \\
=&h\left(X_\msf{R}+\frac{Z_{\msf{B}_1}}{h_{\msf{B_1R}}}\Big|U \right)-h \lp \frac{Z_{\msf{B}_1}}{h_{\msf{B_1R}}} \rp\notag \\ &-h\left(X_\msf{R}+\frac{Z_{\msf{B}_2}}{h_{\msf{B_2R}}}\Big|U\right)+h \lp \frac{Z_{\msf{B}_2}}{h_{\msf{B_2R}}}\rp \label{eq_Gau_R1_UB_maxdiffer} \\
\leq &
\C(\msf{SNR}_{1\msf{R}})-\C(\msf{SNR}_{2\msf{R}}).\label{eq_Gau_R1_UB_maxdiffer1}
\end{align}
Note that the RHS of \eqref{eq_Gau_R1_UB_maxdiffer} is similar to
the achievable rate in a Gaussian wiretap channel. From
\cite{leung1978gaussian} and $|h_{\msf{B_1R}}| \geq
|h_{\msf{B_2R}}|$, \eqref{eq_Gau_R1_UB_maxdiffer1} is valid since
the RHS of \eqref{eq_Gau_R1_UB_maxdiffer} is maximized when
$X_\msf{R}$ is Gaussian conditioned on $U$. Substitute
\eqref{eq_Gau_R1_UB_maxdiffer1} into \eqref{eq_Gau_R1_UB_0}, we
have
\begin{align*}
R_{1}& \leq p\left(r+\C(\msf{SNR}_{1\msf{R}})-\C(\msf{SNR}_{2\msf{R}})\right)+p(1-p)r\\
&=p(2-p)r+p\left(\C(\msf{SNR}_{1\msf{R}})-\C(\msf{SNR}_{2\msf{R}})\right)
\notag
\end{align*}
Together with \eqref{eq_Gau_R2_UB_3}, constraint
\eqref{eq_AWGN_BC_UB_R2} is obtained.

For \eqref{eq_AWGN_BC_UB_R1}, let us give $Y_{\msf{B}_1}$ in
\eqref{eq_Y_B1_neq_gain} to receiver $\msf{B}_2$, we have a
physically degraded channel $X_{\msf{R}} \rightarrow
\{Y_{\msf{B}_1},Y_{\msf{B}_2}\} \rightarrow Y_{\msf{B}_1}$, and
having
\begin{align} \label{eq_Gau_R1_UB_2}
R_{1}&\leq p \; \C(\msf{SNR_{1R}})-pr',
\end{align}
where $r'=I(X_\msf{R};Y_{\msf{B}_1}|U,S_1=1)$. This inequality
follows from steps to reach \eqref{eq_Gau_R2_UB_3}. Also following
steps to reach \eqref{eq_Gau_R1_UB_1}
\begin{align}
R_{2}&\leq
pr'+p^{2}I(X_\msf{R};Y_{\msf{B}_2}|Y_{\msf{B}_1},U,S_1=1,S_{2}=1)\notag
\\ &\;\;\;+p(1-p)(X_\msf{R};Y_{\msf{B}_2}|U,S_{2}=1,S_{1}=0)
\notag\\  & \leq pr'+p^{2}\cdot 0+p(1-p)r' \label{eq_Gau_R2_UB_4}
\\  &=pr'(2-p), \label{eq_Gau_R2_UB_5}
\end{align}
where \eqref{eq_Gau_R2_UB_4} comes from the data processing
inequality,
\begin{align}
I(X_\msf{R};Y_{\msf{B}_2}|U,S_{2}=1)\leq
I(X_\msf{R};Y_{\msf{B}_1}|U,S_{1}=1)=r',
\end{align}
since given $U$ and $S_1=S_2=1$, we have the Markov chain
$X_\msf{R} \rightarrow Y_{\msf{B}_1} \rightarrow Y_{\msf{B}_2}$
from \eqref{eq_Y_B1_neq_gain} and \eqref{eq_Y_B2_neq_gain}. From
\eqref{eq_Gau_R1_UB_2} and \eqref{eq_Gau_R2_UB_5}, we have
\eqref{eq_AWGN_BC_UB_R1}, and then
$\mscr{R}^{\downlink}_{\Out}\lp\mathrm{d}\rp$ is a capacity outer
bound region. Note that for the binary expansion model in Sec.
\ref{subsec_BE_Model}, by giving $Y_{\msf{B}_2}$ to $\msf{B}_1$,
we get
\begin{align}
\textstyle\frac{R_{1}}{p(2-p)}+\frac{R_{2}}{p}\leq
\textstyle\frac{1}{(2-p)}(n_{1\msf{R}}-n_{2\msf{R}})+n_{2\msf{R}}.
\label{eq_BE_BC_UB_R2}
\end{align}
from aforementioned degraded channel arguments. By reversing the
role of $\msf{B}_1$ and $\msf{B}_2$, one get
\begin{align}
\textstyle\frac{R_{1}}{p}+\frac{R_{2}}{p(2-p)} \leq n_{1\msf{R}},
\label{eq_BE_BC_UB_R1}
\end{align}
And the corner point \eqref{eq_BE_BC_UB_point_2} comes from the
intersection of \eqref{eq_BE_BC_UB_R2} and \eqref{eq_BE_BC_UB_R1}.

The outer-bound region $\mscr{R}^{\uplink}_{\Out}\lp\mathrm{d}\rp$
can be proved by allowing users $\msf{A}_1$ and $\msf{A}_2$
cooperate, which is akin to a two transmitter-antennas MISO
channel with per antenna power constraint. This concludes our
proofs for outer bounds of capacity region with delayed state
information $\mscr{C}(\mathrm{d},\mathrm{d})$ in
\eqref{eq_Gau_CapOuter}.

\subsection{Proof for the rest three regions in Theorem \ref{Theo_Main}}
Now we turn to the bounded-gap result for the capacity region with
instantaneous state information for all terminal users and the
relay $\mscr{C}_{\G}(\mathrm{i},\mathrm{i})$. For
\eqref{eq_Gau_CapIn}, to achieve
$\mscr{R}^{\uplink}_{\G,\Inn}\lp\mathrm{i}\rp$ for the uplink
phase, the operations of users and the relay are similar to those
for achieving $\mscr{R}^{\uplink}_{\G,\Inn}\lp\mathrm{d}\rp$ with
delayed state information in Sec \ref{sec:PfInner}. The only
difference is that one can perform on/off power allocation on
\eqref{eq_ul_Lattice} with instantaneous state information. In the
downlink phase, the region
$\mscr{R}^{\downlink}_{\G,\Inn}\lp\mathrm{i}\rp$ is fully
achievable by Gaussian superposition coding with on/off power
allocation. As for the outer-bound regions in
\eqref{eq_Gau_CapOuter}, the cooperative outer bounds for the
uplink with on/off power allocation result in
$\mscr{R}^{\uplink}_{\G,\Out}\lp\mathrm{i}\rp$. From the
uplink-downlink duality (sum power constraint),
$\mscr{R}^{\downlink}_{\G,\Out}\lp\mathrm{i}\rp$ also a capacity
outer-bound region.

Note that in our aforementioned proofs for
$\mscr{C}_{\G}(\mathrm{d},\mathrm{d})$ and
$\mscr{C}_{\G}(\mathrm{i},\mathrm{i})$, the capacity outer and
inner bounds can be decomposed to those for uplink and downlink.
Then these proofs also apply to proving
$\mscr{C}_{\G}(\mathrm{i},\mathrm{d})$ and
$\mscr{C}_{\G}(\mathrm{d},\mathrm{i})$, which establishes the
bounded-gap results for all four $(\mathrm{u},\mathrm{r})$
combinations in Theorem \ref{Theo_Main}.

\bibliographystyle{IEEEtran}
\bibliography{../bibTWR}

\end{document}